\documentclass[11pt,draftcls,onecolumn]{IEEEtran}
\usepackage{amssymb}
\usepackage{cite}
\usepackage{multirow}
\usepackage{tabu}
\usepackage{lineno}
\usepackage{makecell}
\usepackage{xcolor}
\usepackage{tikz}
\usepackage{graphics}
\usepackage{setspace}
\usepackage{threeparttable}
\usepackage{algorithmic}
\usepackage{graphicx}
\usepackage{textcomp}

\ifCLASSOPTIONcompsoc
\usepackage[nocompress]{cite}
\else
\usepackage{cite}
\fi

\usepackage[cmex10]{amsmath}
\usepackage{amsthm}
\usepackage{booktabs}

\newtheorem{theorem}{Theorem}
\newtheorem{lemma}[theorem]{Lemma}

\theoremstyle{definition}%
\newtheorem{definition}{Definition}
\theoremstyle{remark}%
\newtheorem{remark}{Remark}
\newtheorem{example}{Example}

\newcommand{\fq}{\mathbb{F}_q}

\date{}

\renewcommand{\baselinestretch}{1.0}
\title{A Partial-Exclusion Repair Scheme for MDS Codes}

\author{Wei Zhao, Fang-Wei Fu and Ximing Fu
	\thanks{Wei Zhao is with the School of Mathematics, Foshan University, Guangdong, 528000, China	(e-mail: zhaowei@fosu.edu.cn; fuzicong@sina.cn).}
	\thanks{Fang-Wei Fu is with Chern Institute of Mathematics and LPMC, Nankai University, Tianjin, 300071, China  (email: fwfu@nankai.edu.cn).}
	\thanks{Ximing Fu is with the School of Computer Science and Technology, Harbin Institute of Technology, Shenzhen, China (email:fuximing@hit.edu.cn).}  
}

\begin{document}

\maketitle

\begin{abstract}
	For scalar maximum distance separable (MDS) codes, the conventional repair schemes that achieve the cut-set bound with equality for the single-node repair have been proven to require a super-exponential sub-packetization level.
	As is well known, such an extremely high level severely limits the practical deployment of MDS codes.
	To address this challenge, we introduce a partial-exclusion (PE) repair scheme for scalar linear codes. 
	In the proposed PE repair framework, each node is associated with an exclusion set. 
	The cardinality of the exclusion set is called the flexibility of the node. 
	The maximum value of flexibility over all nodes defines the \textit{flexibility} of the PE repair scheme. Notably, the conventional repair scheme is the special case of
	PE repair scheme where the flexibility is 1.

	Under the PE repair framework, for any valid flexibility, we establish a lower bound on the sub-packetization level of MDS codes that meet the cut-set bound with equality for single-node repair.
	This bound reveals an intrinsic trade-off among  the flexibility, the smallest possible sub-packetization level, and the minimum normalized repair bandwidth (\textit{i.e.}, the minimum number of bits transmitted per repaired bit). 
	By plotting the trade-off curve of MDS codes, it is shown that the PE repair scheme with the intermediate flexibility $t$ attains better trade-off compared to the two endpoints (\textit{i.e.}, the conventional case $t=1$ and the naive repair approach $t=n-k$).
	To realize MDS codes attaining the cut-set bound under the PE repair framework, we propose two generic constructions of Reed-Solomon (RS) codes. Moreover, we demonstrate that for a sufficiently large flexibility, the sub-packetization level of our constructions is strictly lower than the known lower bound established for the conventional repair schemes. 
	This implies that, from the perspective of sub-packetization level, our constructions outperform all existing and potential constructions designed for conventional repair schemes. 
	To validate this superiority, we explicitly construct a $(12,8)$ RS code with flexibility 3 via the first construction. 
	This code achieves a sub-packetization level 2310, which is substantially less than both the known lower bound of 510510 established for the conventional repair scheme and
	the level of up to $3\times 10^{14}$ required by prior constructions.
	Our second construction achieves even lower sub-packetization level, as exemplified by an explicit $(17,9)$ RS code with flexibility 7 and sub-packetization level 30. This level is drastically less than the known lower bound of 9699690 established for the conventional repair scheme and
	the prior achievable level of up to $2.75\times 10^{30}$.
	Finally, we implement the repair process for these codes as executable Magma programs, thereby exhibiting the practical efficiency of our constructions.
\end{abstract}
\section{Introduction}
Node failures are inevitable in large-scale distributed storage systems (DSSs). 
Maximum distance separable (MDS) codes constitute one of the most widely used families of erasure codes in DSSs \cite{dinh2022practical}. According to measurements from Facebook's production warehouse cluster, approximately 98.08\% of failure incidents involve a single failed node on average \cite{rashmi2013solution}. The problem of recovering the data stored in the failed node exactly is referred to as the \textit{(exact) repair problem}.

Guruswami and Wootters~\cite{guruswami2017} introduced a definition of the \textit{(linear)} repair scheme for a scalar linear code. In their approach, the erased symbol is recovered using the sub-symbols derived from the coded symbols, rather than using the coded symbols themselves. The field to which the sub-symbols belong is referred to as the \textit{base field} of a code. The \textit{repair bandwidth} is defined as the total number of symbols in the base field required to be transmitted to repair the failed nodes.
The \textit{sub-packetization level} of a code is referred to the extension degree of the symbol field over the base field. They demonstrated that, for (scalar) MDS codes, the explicit repair scheme they proposed can yield a strictly lower repair bandwidth compared to the naive repair method.

This work has stimulated extensive subsequent research on the repair problem of MDS codes. In this line of research, an MDS code construction referred to encompass a family of codes together with a specific repair scheme designed for them. The \textit{cut-set bound} establishes the minimum possible repair bandwidth for recovering failed nodes.
Moreover, we say that a construction or a code achieves (resp., asymptotically approaches) the cut-set bound if the repair bandwidth of its equipped repair scheme meets the cut-set bound with equality (resp., approaches the cut-set bound as the code length tends to infinity). 
If the equipped repair scheme of a construction or a code is designed according to the repair scheme definition given in \cite{guruswami2017} or other repair scheme definitions, we term that the construction or the code is \textit{under such a repair framework}.


For single-node repair of MDS codes, on one hand, considerable efforts have been devoted to constructing MDS codes with non-trivial repair schemes over a reasonably sized symbol field \cite{Berman2022RepairingRC,guruswami2017,Li2019OnTS,Xu2024CooperativeRO}. It should be noted that these constructions do not achieve the cut-set bound.
On the other hand, other research has focused on constructing MDS codes that achieve or asymptotically approach the cut-set bound, albeit at the cost of requiring an extremely large symbol field, see  \cite{Chowdhury2021ImprovedSF,ye2016explicit} for asymptotically optimal and \cite{tamo2017optimal, tamo2019repair} for optimal constructions.
For multiple-node repair of MDS codes, in
general, there are two models called centralized model \cite{cadambe2013asymptotic, rawat2018centralized, ye2017explicit, zorgui2017centralized} 
and cooperative model \cite{balaji2018erasure,ford2010availability,hu2013analysis,shum2013cooperative, vajha2018clay}. Another research direction focuses on the repair of MDS codes in the rack-aware storage model \cite{chen2019explicit,wang2023low,wang2023rack,yang2024transformation}.

For MDS codes achieving the cut-set bound under the conventional repair framework, it has been established in \cite{tamo2017optimal} that the asymptotic lower bound of sub-packetization level $L$ of such codes is $L\geq e^{(1+o(1))k\log k}$.
This lower bound reveals that the current repair framework for construction of MDS codes achieving the cut-set bound necessarily entail extremely high sub-packetization level. It is well recognized that such excessive sub-packetization level not only imposes significant challenges for practical implementation, but also restricts the minimum size of a file that can be stored. These limitations motivate the need for further investigation. In this paper, we focus on the scenario of the single-node repair.

\subsection{Our Results}
\textbf{1). A PE repair scheme.} We introduce a partial-exclusion (PE) linear repair scheme for $(n, k)$ scalar linear codes (see Definition~\ref{definition}).
For a PE repair scheme, each node is assigned an exclusion set. 
Throughout the paper we use the notation $[n]:=\{1,2,\ldots,n\}$.
For any failed node $i\in [n]$ with the associated exclusion set $A_i\subseteq [n]$, 
the helper nodes for repair must then be chosen from the set of surviving nodes outside $A_i$. 
Consequently, the allowable number $d_i$ of helper nodes for repairing the node $i$ is $k\leq d_i\leq n-t_i$, where $t_i$ is the cardinality of $A_i$
satisfies $1\leq t_i\leq \min\{k,n-k\}$. 
The local parameter $t_i$ for each $i\in [n]$ is referred to as the \textit{flexibility} of the node $i$. The maximum value among all $t_i$, denoted by $t$, is defined as the \textit{flexibility} of the PE repair scheme. 
In general, the values of $t_i$ may differ among different nodes. 
Note that a PE repair scheme reduces to the conventional repair scheme when $t_i=1$ for all $i\in[n]$, \textit{i.e.}, $t=1$. 
For any valid flexibility, we establish a lower bound of sub-packetization level for scalar MDS codes achieving the cut-set bound for the single-node repair under the PE repair framework (see Theorem~\ref{our_bound} and Theorem~\ref{cor_bound}). This bound is significantly reduced when the flexibility is increased.
%
%
%

\textbf{2). A trade-off.} There exists a trade-off governed by the flexibility of PE repair schemes.
For the sake of clarity and illustration, we assume a uniform value $t_i=t$ for all $i\in[n]$.
%
To precisely characterize the trade-off, 
for a repair scheme with repair bandwidth $\beta$ and sub-packetization level $L$, we define the \textit{normalized repair bandwidth} of the scheme as $\bar{\beta} = \beta / L$, representing the number of bits transmitted per repaired bit.
The \textit{minimum normalized repair bandwidth} for a PE repair scheme refers to the normalized repair bandwidth of repairing a failed node by accessing the maximum allowable number of helper nodes, which is $n-t$ for a given flexibility $t$. The \textit{smallest possible sub-packetization level} for MDS codes achieving the cut-set bound for the single-node repair
is defined as the lower bound of sub-packetization level for such MDS codes.
In section~\ref{tradeoffsec}, for MDS codes under the PE repair framework, we demonstrate that there exists a trade-off among the smallest possible sub-packetization level, the minimum normalized repair bandwidth, and the flexibility. 

For the $(14,10)$~RS codes over the base field $\mathbb{F}_2$, we present the corresponding trade-off curve (see Fig.~1). 
Notably, from the viewpoint of repair bandwidth, $t = 1$ corresponds to conventional repair schemes whose repair bandwidth matching the $(14,10)$ minimum storage regenerating (MSR) codes, while $t = 4$ corresponds to the naive repair approach for general MDS codes.
Fig.~1 shows that in practical scenarios, $t = 2$ and $t=3$ achieve a better trade-off than the two endpoints ($t=1$ and $t=4$). This is primarily because $t = 1$ requires an impractically high sub-packetization level ($2.23\times 10^{8}$) while offering only a marginal gain in the minimum normalized repair bandwidth. Meanwhile, $t=4$ requires a relatively high repair overhead for the single-node repair.
\begin{figure}[ht]%
	\centering
	\label{tradeoff_fig}
	\includegraphics[width=0.7\textwidth]{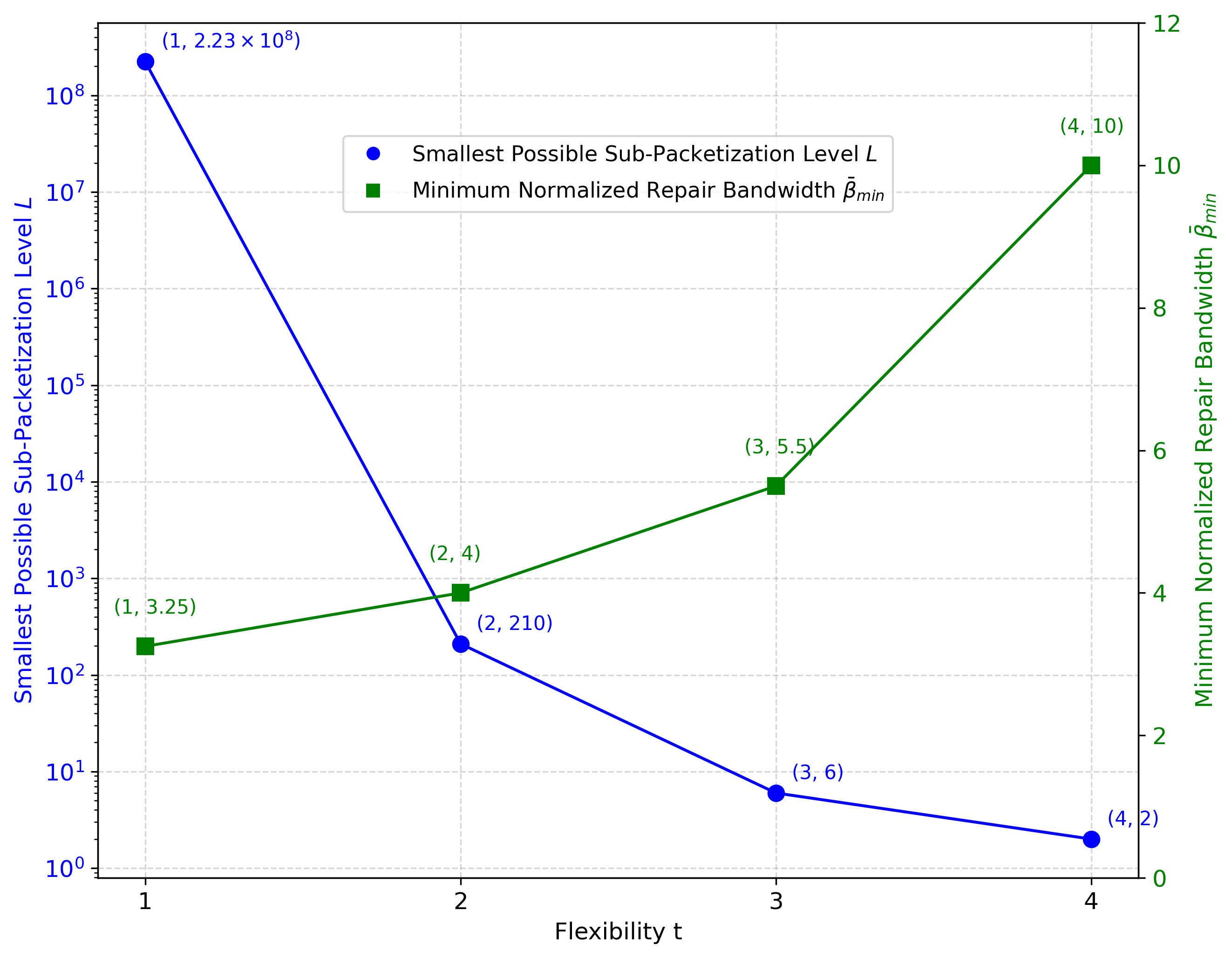}
	\caption{The trade-off between flexibility $t$, the smallest possible sub-packetization level $L$ and the minimum normalized repair bandwidth $\bar{\beta}_{min}$, for the $(14,10)$ RS codes over the base field $\mathbb{F}_2$.}
\end{figure}

\textbf{3). A generic construction of RS codes attaining the cut-set bound under the PE repair framework.}
For any positive integers $n\geq k\geq 1$ and $1\leq t\leq \min\{k,n-k\}$, we construct a class of $(n,k)$ RS codes that attain the cut-set bound under the PE repair framework for any valid flexibility $t$ (see Construction~1). When $t\geq \frac{n}{k-3}$, the sub-packetization level of the constructed RS codes is even strictly less than the sub-packetization lower bound established for the conventional repair schemes ($t=1$) \cite{tamo2019repair}. 
This indicates that, from the perspective of sub-packetization level, for RS codes achieving the cut-set bound, our construction outperforms all the existing and potential constructions designed for conventional repair schemes.
Given that $(12,8)$ RS codes are typical RS codes that adopted in industry Cloud Storage \cite{lai2015atlas}. 
To validate the superiority of our construction, we provide an explicit $(12,8)$ RS code over the base field $\mathbb{F}_{2}$ using the proposed construction (see Example~\ref{example1}).
Accordingly, Table~\ref{comparison} presents a comparison between the explicit $(12,8)$ RS codes derived from our construction and existing constructions that achieve or asymptotically approach the cut-set bound.
It is worth noting that our proposed code achieves the sub-packetization level 2310, while the minimum possible sub-packetization level for $(12,8)$ RS codes achieving the cut-set bound under conventional repair framework is $510510$. Meanwhile, the existing constructions designed for conventional repair schemes with the same normalized repair bandwidth as ours require a sub-packetization level approximately $ 3.04\times 10^{14}$.
This prohibitively high level pose a fundamental barrier to practical deployment.
Additionally, to verify the efficiency of our construction, we implement the repair program of the constructed $(12,8)$ RS code using Magma software on a regular commercial computer. 

\textbf{4). The second generic construction of RS codes attaining the cut-set bound under the PE repair framework.}
To further reduce the sub-packetization level of RS codes that achieves the cut-set bound under the PE repair framework, a new class of RS codes is presented with a relatively stringent constraint on $t$ (see Construction 2).
Accordingly, the sub-packetization level of the second construction can also strictly less than the sub-packetization lower bound for conventional repair schemes (see Remark~\ref{sp_compare_2}).
Furthermore, we provide an explicit $(17,9)$ RS code over the base field $\mathbb{F}_{4}$ 
with sub-packetization level $30$ and flexibility $7$ and detail its repair scheme  (see Example~\ref{example2}). 
It is worth noting that the minimum possible sub-packetization level for $(17,9)$ RS codes achieving the cut-set bound under the conventional repair framework is $9699690$. Meanwhile, as shown in Table~\ref{comparison}, even when the average minimum normalized repair bandwidth per node of our codes is set to be slightly less than the normalized repair bandwidth of existing constructions, the existing constructions designed for conventional repair schemes require a sub-packetization level approximately $ 2.75\times 10^{30}$.
Additionally, we also implemented the repair scheme of our proposed $(17,9)$ RS code using Magma software.
	\begin{table}[t]
	\centering
	\caption{Constructions of Explicit RS codes with optimal or asymptotically optimal repair bandwidth}
	\resizebox{\textwidth}{!}{
		\begin{tabular}{ccccccc}
			\toprule[0.8pt]
			\multirow{2}{*}{Constructions}  & \multirow{2}{*}{$(q,n,k)$} &			
			Sub-Packetization &Repair Bandwidth  & Repair & Normalized & Repair Meets \\
			& & Level &  (Bits) $^{\mathrm{a}}$&  Locality $^{\mathrm{b}}$ & Repair Bandwidth$^{\mathrm{c}}$ &  Cut-Set Bound\\
			\midrule
			\cite{Chowdhury2021ImprovedSF} & $(2,12,8)$ &  $8192$ & $47104$  & 11 & 5.75 &  Asymptotically \\
			\cite{ye2016explicit} & $(2,12,8)$ &  $16777216$ & $54525952$  &  11 & 3.25 &  Asymptotically \\
			\cite{tamo2017optimal,tamo2019repair} & $(2,12,8)$ &  $\approx3.04\times 10^{14}$ & $\approx 1.36\times 10^{15}$  &  9 & 4.5 &  Yes \\
			Construction 1 (Sec. IV) & $(2,12,8)$ &  $2310$ & $10395$  &  9 & 4.5 & Yes \\
			\hline
			\cite{tamo2017optimal,tamo2019repair} & $(4,17,9)$ &  $\approx 2.75\times 10^{30}$ & $\approx 2.02\times 10^{31} $  & 11 &  $11/3$ & Yes \\
			& & &  Nodes in $A_1$: 300 $^{\mathrm{d}}$& 10 & 5 &\\
			\multirow{-1}{*}{Construction 2 (Sec. V)} & \multirow{-1}{*}{$(4,17,9)$} & \multirow{-1}{*}{$30$} &Nodes in $A_2$: 220 \,\,\, & 11 & 11/3 & \multirow{-1}{*}{Yes}\\
		     & & &   Nodes in $A_3$: 156 \,\,\, & 13 & 2.6 &\\
			\bottomrule[0.8pt]\\
		\end{tabular}
	}
	\footnotesize
	\raggedright
	\textbf{Note:} a. The repair bandwidth column quantifies the total number of bits required to repair an erased coded symbol;
	b. The repair locality of a node is defined as the number of helper nodes accessed for the repair of the node;
	c. The normalized repair bandwidth is defined as the number of bits transmitted per repaired bit, given by $\frac{d}{d-k+1}$, where $d$ is the repair locality and $k$ is the dimension of codes;
	d. The nodes in the $(17,9)$ RS code via Construction 2 are partitioned into three disjoint groups, $A_1, A_2$ and $A_3$, where each group requires a different repair locality. Consequently, the repair bandwidth differs across the groups.
	\label{comparison}
\end{table}

\subsection{Organization of the Paper}
The rest of this paper is organized as follows. In Section \ref{section:2}, we review some notations and foundational results. In Section~\ref{repair}, we introduce a PE repair scheme and establish a lower bound of sub-packetization level of scalar MDS codes attaining the cut-set bound for the single-node repair under the PE repair framework. In Section~\ref{tradeoffsec}, we analyze the trade-off among the smallest possible sub-packetization level, the minimum normalized repair bandwidth, and the flexibility for scalar MDS codes under the PE repair framework.
In Section~\ref{construction1} and Section~\ref{construction2}, we present two generic constructions of RS codes that meet the cut-set bound with equality for the single-node repair under the PE repair framework. Conclusions is provided in Section~\ref{conclusion}.

\section{Preliminaries}\label{section:2}
\subsection{Reed-Solomon Codes}
Let $q$ be a prime power, $\fq$ be the finite field of $q$ elements. Denote the multiplicative group of the finite field $\fq$ as $\fq^{*}$. 
An $(n,k)$ (scalar) linear code $\mathcal{C}$ is defined as an $\fq$-linear subspace of $\fq^n$ with dimension $k$.
The \textit{dual code} $\mathcal{C}^{\perp}$ of $\mathcal{C}$ is the subspace of $\fq^n$ given by
\[\left\{(c_1,\ldots,c_n)\in \fq^n:\,\sum_{i=1}^{n}b_ic_i=0 \text{ for all } (b_1,\ldots,b_n)\in \mathcal{C}\right\}.\]
For two codewords $x=(x_1,\ldots,x_n)$ and $y=(y_1,\ldots,y_n)$ of $\mathcal{C}$, the \textit{(Hamming) distance} between $x$ and $y$ is defined as the number of coordinates in which they differ, \textit{i.e.},
$d_H(x,y):=|\{1\leq i\leq n :\, x_i\neq y_i\}.$
The \textit{minimum (Hamming) distance} of $\mathcal{C}$ is defined as
$d_H(\mathcal{C}):=\min_{x,y\in \mathcal{C}, \, x\neq y} d_H(x,y).$
An $(n,k)$ linear code $\mathcal{C}$ over a finite field $\fq$ with minimum distance $d_H$ is called a \textit{maximum distance separable (MDS) code} if it achieves the Singleton bound with equality, \textit{i.e.}, $d_H=n-k+1$.

Let $A=\{\alpha_1,\alpha_2,\ldots,\alpha_n\}$ be a subset of $\fq$ consisting of $n$ distinct elements in $\fq$. For a positive integer $k$ with $1\leq k\leq n$, the \textit{generalized Reed-Solomon code} $GRS(n,k,A,\mathbf{v})$ is defined as
\[\Bigl\{\bigl(v_1f(\alpha_1),\ldots,v_nf(\alpha_n)\bigr)\,:\,f(x)\in \fq[x],\, \deg(f)<k\Bigr\},\]
where $\mathbf{v}=(v_1,v_2,\ldots,v_n)\in\left(\fq^{*}\right)^{n}$, and $\deg(f)$ denotes the degree of $f(x)$. If $\mathbf{v}=(1,1,\ldots,1)$, then the GRS code is referred to as an \textit{Reed-Solomon (RS) code} and denoted as $RS(n,k,A)$.
It is well known that $GRS(n,k,A,\mathbf{v})$ is an MDS code with dimension~$k$. The dual code of $GRS(n,k,A,\mathbf{v})$ is a $GRS(n,n-k,A,\mathbf{v}^{*})$ for some vector $\mathbf{v}^{*}=(v_1^{*},v_2^{*},\ldots,v_n^{*})\in\left(\fq^{*}\right)^{n}$ \cite{lidl97}.
Consequently, for any polynomials $f(x),\, g(x) \in \fq[x]$ with degrees strictly less than $k$ and $n-k$, respectively, it holds that $$\sum_{i=1}^{n}v_i v_i^{*} f(\alpha_i)g(\alpha_i) = 0.$$ Such polynomials $f(x)$ and $g(x)$ are referred to as the \textit{encoding polynomial} and the \textit{parity-check polynomial} of $GRS(n,k,A,\mathbf{v})$, respectively.

Let $\mathbb{F}_{q^L}$ be the extension field of degree $L$ over $\fq$. 
The \textit{trace map} from $\mathbb{F}_{q^L}$ onto $\fq$
is defined as
\begin{equation*}
	Tr(\alpha):=\alpha+\alpha^q+\cdots+\alpha^{q^{L-1}}\in \fq,\,\alpha\in \mathbb{F}_{q^L},
\end{equation*}
which is an $\fq$-linear surjection.
For any positive divisor $m$ of $L$, the finite field $\mathbb{F}_{q^m}$ is a subfield of $\mathbb{F}_{q^L}$, and $\mathbb{F}_{q^L}$ may be viewed as a vector space of dimension $\frac{L}{m}$ over $\mathbb{F}_{q^m}$.
Considering $\mathbb{F}_{q^L}$ as an $L$-dimensional $\fq$-space, and suppose that $\{\alpha_1,\alpha_2,\ldots,\alpha_L\}$ is a basis of $\mathbb{F}_{q^L}$ over $\fq$. It follows from \cite[Page 58]{lidl97} that there exists a unique basis $\{\beta_1,\beta_2,\ldots,\beta_L\}$ of $\mathbb{F}_{q^L}$ over $\fq$ satisfying that
\begin{equation}
	Tr(\alpha_i\beta_j)=\left\{ \begin{matrix}
		1, \text{ \, for } i=j, \\
		0, \text{ \, for } i\neq j,  \\
	\end{matrix} \right.
\end{equation}
for all $1\leq i,j\leq L$. We refer to the set $\{\beta_1,\beta_2,\ldots,\beta_L\}$ as the \textit{(trace) dual basis} of  $\{\alpha_1,\alpha_2,\ldots,\alpha_L\}$. If two bases $\{\alpha_1,\alpha_2,\ldots,\alpha_L\}$ and $\{\beta_1,\beta_2,\ldots,\beta_L\}$ are a pair of dual bases of $\mathbb{F}_{q^L}$ over $\fq$, then for any $\alpha\in \mathbb{F}_{q^L}$, it holds that
\begin{equation}\label{dual_bases}
	\alpha=\sum_{j=1}^{L}Tr(\beta_j\alpha)\alpha_j.
\end{equation}
So each coordinate of $\alpha$ in the basis of $\{\alpha_1,\alpha_2,\ldots,\alpha_L\}$ is obtained by applying the trace map to the product of $\alpha$ with the corresponding dual basis element. 
Dual bases are widely used in the design of repair schemes for RS codes \cite{Berman2022RepairingRC,guruswami2017,Li2019OnTS,Xu2024CooperativeRO,Chowdhury2021ImprovedSF,ye2016explicit,tamo2017optimal}.

\subsection{Repair schemes for MDS codes}

Consider an $(n, k, L)$ (linear) array code $\mathcal{C}$ over a finite field $F$. Specifically, $\mathcal{C}$ is a collection of codewords $c= (c_1, \dots, c_{n})$, where $c_j = (c_{j,1}, \dots, c_{j,L})^T \in F^L$ for $j = 1, \dots, n $, and $\mathcal{C}$ forms a $(kL)$-dimensional linear subspace over $F$.  
The parameter $L$ is termed the \textit{sub-packetization level} of~$\mathcal{C}$.
Furthermore, if $\mathcal{C}$ is strengthened to be linear over $F^L$, then $\mathcal{C}$ is referred to as a \textit{scalar} code. Moreover, $\mathcal{C}$ is maximum distance separable (MDS) if each codeword in $\mathcal{C}$ can be recovered from any $k$ of its coordinates. 
For any codeword $c=(c_1,c_2,\ldots,c_n)$ of $\mathcal{C}$, where $c_i\in F^L$, the symbols of $c_i$ are stored in node $i$ for each $i\in [n]$.
We do not distinguish between nodes and the coordinates of a codeword.

Denote by $\mathcal{R} \subset [n]\setminus\{i\}$ the set of $d= |\mathcal{R}|$ helper nodes, where $k \leq d \leq n-1$. 
Suppose the failed node has index $i$. To recover it, we utilize the symbols from the helper nodes $c_e$ ($e \in \mathcal{R}$).
A repair scheme with $d$ helper nodes is formed of $d$ functions $f_t : F^{L} \to F^{\beta_t}$, $t= 1, \dots, d$ and a function $g : F^{\sum_{t=1}^{d} \beta_t}  \to F^{L}$. For each $e_t \in \mathcal{R}$ the function $f_t$ maps $c_{e_t}$ to some $\beta_t$ symbols of $F$. The function $g$ accepts these symbols from all the helper nodes as arguments, and returns the values of the failed nodes $c_{i}=g(\{f_t(c_{e_t}), e_t \in \mathcal{R}\})$
for all $c \in \mathcal{C}.$ If the function $f_t,g$ are $F$-linear, the repair scheme is said to be \textit{linear}. 
The quantity $\beta(\mathcal{R}) = \sum_{e \in \mathcal{R}} \beta_e$ is called the \textit{repair bandwidth} of recovering the failed node $i$ from the helper nodes in $\mathcal{R}$. The repair bandwidth of the scheme is defined as the maximum of $\beta(\mathcal{R})$ over all possible node failures and selections of helper nodes.
The cut-set bound \cite{cadambe2013asymptotic,dimakis2010network} shows that the repair bandwidth is lower bounded by
\begin{equation}\label{cut-set-bound}
	\beta(\mathcal{R})\geq\frac{|\mathcal{R}|L}{|\mathcal{R}|-k+1}.
\end{equation} 

The repair methodology for array codes also applies to scalar codes. To be precise, in a DSS employing a scalar linear $(n,k)$ MDS code $\mathcal{C}$ over a finite field $E$, the stored file is divided into $k$ blocks, each represented as an element of $E$.
These $k$ elements of $E$ are encoded into a codeword consisting of $n$ symbols of $E$ using $\mathcal{C}$, which are then distributed across $n$ nodes.
While a failed node can be recovered by downloading all data from any $k$ surviving nodes due to the MDS property, this naive repair strategy incurs prohibitively high repair bandwidth, limiting the practical applicability of scalar MDS (particularly RS) codes in DSS. To reduce bandwidth, one may regard $E$ as an $L$-dimensional vector space over a subfield $F$ of $E$, 
where $L=\log_{|F|} |E|$. This allows the scalar linear $(n,k)$ MDS code over $E$ to be seen as a linear $(n, k, L)$ array code over $F$.
Following the array code paradigm, a repair scheme may download only partial content (subsymbols over $F$) from helper nodes. The \textit{repair bandwidth} is the total number of
$F$-symbols downloaded in the worst case. Crucially, the cut-set bound \eqref{cut-set-bound} also applies to scalar MDS codes under this subsymbol-based repair model.

\section{A Partial-Exclusion Linear Repair Scheme}\label{repair}
Guruswami and Wootters conducted seminal research on the repair problem of scalar linear codes and introduced the concept of the linear repair scheme \cite{guruswami2017}. 
The repair scheme proposed in this paper is slightly different from that scheme. To better align with the conventional repair scheme, we introduce the definition of a partial-exclusion repair scheme by adding the property of flexibility into the conventional repair scheme. 
Following the treatment in \cite{guruswami2017}, a linear code is regarded as a set of functions evaluated at a set of evaluation points.
A code $\mathcal{C}$ of length $n$ over a field $E$ is a subset $\mathcal{C}\subseteq E^n$.
Let $\mathcal{F}$ be a collection of functions defined on $E$. Let $A\subseteq E$ be a set of evaluation points.
The code $\mathcal{C}\subseteq E^n$ determined by $\mathcal{F}$ and $A$ is 
\[\mathcal{C}=\{(f(\alpha_1),f(\alpha_2),\ldots,f(\alpha_n))\,:\,f\in \mathcal{F}\}.\]
In particular, the positions $1,\ldots,n$ of a codeword are indexed as the evaluation points
$\alpha_1,\ldots,\alpha_n$ in the following definition. 

\begin{definition}[A Partial-Exclusion Linear Exact Repair Scheme]\label{definition}
	Let $\mathcal{C}$ be a scalar linear code over $E$ of length $n$ and dimension $k$, given by a collection of function $\mathcal{E}$ and a set of evaluation points $A=\{\alpha_1,\ldots,\alpha_n\}$. 
	For each $i\in [n]$, the node $\alpha_i$ is associated with an exclusion set $A_i$ such that $\alpha_i\in A_i$ and the cardinality of $A_i$ is $t_i$, where $1\leq t_i\leq \min\{n-k,k\}$.
	A partial-exclusion (PE) linear exact repair scheme $\mathcal{S}$ for $\mathcal{C}$ over a subfield $K \leq E$ consists of the following.
	\begin{itemize}
		\item For each $\alpha_i \in A$, and each $\alpha \in A \setminus A_i$, a set of queries $Q_\alpha(\alpha_i) \subseteq E$.
		\item For each $\alpha_i \in A$, a linear reconstruction algorithm that computes
		\[
		f(\alpha_i) = \sum_{j=1}^{L} \lambda_j v_j
		\]
		for coefficients $\lambda_j \in K$ and a basis $v_1, \ldots, v_L$ for $E$ over $K$, so that the coefficients $\lambda_j$ are $K$-linear combinations of the queries
		\[
		\bigcup_{\alpha \in A \setminus A_i} \Bigl\{Tr_{E/K}(\gamma f(\alpha)) : \gamma \in Q_\alpha(\alpha_i) \Bigr\}.
		\]
	\end{itemize}
	The repair bandwidth $\beta$ of the failed node $\alpha_i$ of the exact repair scheme is the total number of sub-symbols in $K$ returned by each node $\alpha \in A\setminus A_i$:
	\[
	\beta = \max_{\alpha_i \in A} \sum_{\alpha \in A \setminus A_i} \bigl| Q_\alpha(\alpha_i) \bigr|.
	\]
	For each $\alpha_i \in A$, the \textit{ flexibility} of node $\alpha_i$ is the cardinality of the associated exclusion set $A_i$:
	\[t_i=\bigl| A_i \bigr|.\]
	The \textit{ flexibility} of the exact repair scheme is the maximum value of flexibility of nodes $\alpha_i$ for all $i\in [n]$: 
	\[
	t=\max_{i\in[n]} t_i.
	\]
	The repair locality $d_{i}$ for node $\alpha_i$ is the number of $\alpha$ which are required to respond:
	\[
	d_{i} = \sum_{\alpha \in A \setminus  A_i} 
	\mathbf{1}_{Q_\alpha(\alpha_i) \neq \emptyset},
	\]
	where $k\leq d_{i}\leq n-t_i$.
	The \textit{repair locality} $d$ of the exact repair scheme is the maximum value of $d_{i}$ for all $i\in [n]$:
	\[
	d = \max_{i \in [n]} d_{i},
	\]
	where $k\leq d\leq n-t$.
	Let
$
	L = \log_{|K|}(|E|)
$ be the dimension of $E$ as a vector space over $K$. Thus, we can view each symbol from $E$ as a vector of $L$ sub-symbols from $K$. We refer to the fields $E$ and $K$ as the \textit{symbol field} and the \textit{base field} of $\mathcal{C}$, respectively. The value $L$ is referred to as the \textit{sub-packetization level} of $\mathcal{C}$.
\end{definition}

\begin{remark}[Repair Locality]
	It is important to note that the PE repair scheme considered in this paper differs from the conventional repair scheme in a subtle yet significant manner. Conventional repair schemes typically assume a uniform allowable range of the repair locality $d_i$ for all nodes $i$ which is $k\leq d_i\leq n-1$ (see \cite{Chowdhury2021ImprovedSF,ye2016explicit,tamo2017optimal,tamo2019repair} for example). As has been demonstrated, this conventional setting forces a substantially high sub-packetization level if the repair bandwidth of such repair scheme achieves the cut-set bound for the single-node repair. In contrast, our scheme allows the range of the repair locality $d_i$ ($k \leq d_i \leq n - t_i$) to vary for each $i\in [n]$. The PE repair framework reduces to the conventional one in the case that $t_i=1$ for all $i\in[n]$, \textit{i.e.}, the flexibility $t=1$. We claim that the flexibility of PE repair schemes plays a crucial role in reducing the sub-packetization level of RS codes that achieve the cut-set bound.
\end{remark}

\begin{remark}
	To comprehend the intention of modifying the definition of the conventional repair scheme, one must recognize a certain fact.
	If a failed node can be recovered with the assistance of $d < n - 1$ nodes, then it can definitely be repaired with the help of $d + 1$ nodes. 
	However, if a failed node can be recovered with the help of $d < n - 1$ nodes and the repair bandwidth can reach the cut-set bound, it cannot be inferred that the repair with $d + 1$ helper nodes can also reach the cut-set bound.
	Specifically, the PE repair schemes focus on repair scenarios in which the repair locality $d_{\alpha_i}$ satisfy $k\leq d_{\alpha_i} \leq n - t_i$ when repairing the node $\alpha_i\in A$. 
	This is a design focus, rather than a restriction on recovery capability.	 
	Therefore, when we consider the PE repair schemes in the context of MDS codes achieving the cut-set bound, the design intention of the PE repair schemes becomes apparent.
\end{remark}

In general, the parameters $t_i$ ($1\leq i\leq n$) are not necessarily equal. 
We consider a specific situation where the evaluation set $A$ is divided into $l$ pairwise-disjoint subsets $A_1, \ldots, A_l$, where $1<l\leq n$. For each $i\in [l]$, the associated exclusion set of every node in $A_i$ is exactly $A_i$. 
Unless otherwise specified, all the PE repair schemes considered in this paper are confined to this scenario.
Thus, for a specific PE repair scheme $\mathcal{S}$, there exist the corresponding pairwise disjoint exclusion sets $A_1, \ldots, A_l$.

In the following, we demonstrate that for any scalar linear MDS code that meets the cut-set bound with equality under the PE repair framework, the lower bound of sub-packetization level is determined by the flexibility of nodes.
To simplify notation, we slightly overload the notation of exclusion sets $A_1, \ldots, A_l$, which can represent either the subsets of the evaluation set $ \{\alpha_1,\ldots,\alpha_n\}$ or the subsets of the index set $\{1,\ldots,n\}$. The specific meaning will be clear from context.

Let $n$, $k$, and $l$ be positive integers satisfying $n > k \ge 1$ and $n \ge l \ge 2$, and let $\fq$ be a finite field with an extension field $\mathbb{F}_{q^L}$.
Let $\mathcal{C}$ be an $(n,k)$ scalar linear MDS code over the symbol field $\mathbb{F}_{q^L}$ with the base field $\fq$.
Suppose the node set $[n]$ of $\mathcal{C}$ is partitioned into $l$ pairwise disjoint subsets $A_1, \ldots, A_l$, where each subset $A_i$ has size $t_i$ and $1 \le t_i \le \min\{n - k, k\}$ for all $ 1\leq i \leq l $.
Let $\mathcal{S}$ be a PE repair scheme of $\mathcal{C}$ with associated exclusion sets $A_1, \ldots, A_l$.
Therefore, each node from $A_i$ can be repaired by downloading data from the nodes indexed by $\mathcal{R}_i$, where $\mathcal{R}_i\subseteq [n]\setminus A_i$. 

Without loss of generality, due to the equivalence of linear codes, we assume throughout that 
$t_1\leq t_2\leq \cdots \leq t_l$, and accordingly define these pairwise disjoint subsets explicitly as:
\begin{equation}
	A_1 = \{1,\ldots,t_1\},\quad A_i = \left\{ \sum_{j=1}^{i-1} t_j + 1, \ldots, \sum_{j=1}^{i} t_j \right\} \text{ for } i = 2, \ldots, l.
\end{equation}
Let
\begin{equation}\label{ti}
	T_1=1,\qquad T_i=\sum_{j=1}^{i-1}t_j+1 \text{ for } i = 2, \ldots, l.
\end{equation}
Thus, the node indexed by $T_i$ corresponds to the first node in the group $A_i$ for $i\in [l]$.
Let $w$ be the greatest possible integer such that $\sum_{i=1}^{w}t_i\leq k$. For each $j=1,2,\ldots,w$,
the repair scheme of $\mathcal{C}$ for repairing the node $T_j$ shows that there exist $L$ codewords $$c_i=(c_{i,1},c_{i,2},\ldots,c_{i,n})\in \mathcal{C}^{\perp},$$
for $i=1,2,\ldots,L$
such that 
\begin{equation}\label{dim}
	\dim_{\fq}(c_{1,T_j},c_{2,T_j},\ldots,c_{L,T_j})=L
\end{equation}
and
\begin{equation}\label{dim2}
	\beta_{T_j}=\sum_{i\in\mathcal{R}_j} \dim_{\fq}(c_{1,i},c_{2,i},\ldots,c_{L,i})=\frac{(n-t_j)L}{n-t_j-k+1},
\end{equation}
where $\mathcal{R}_{j}\subseteq [n]\setminus A_j$ denotes the index set of helper nodes for repairing node $T_j$ and $\beta_{T_j}$ denotes the repair bandwidth for repairing node $T_j$. 
For each $j=1,2,\ldots,w$, we may assume that $\mathcal{R}_{j}= [n]\setminus A_j$ by puncturing the code to the coordinates corresponding to the helper nodes.
Let $r=n-k$ and $H\in E^{r\times n}$ be a parity-check matrix of $\mathcal{C}$.
Since $H$ is also the generator matrix of the dual code of $\mathcal{C}$, it follows that there exists an $L\times r$ matrix 
\begin{equation}\label{btj}
	B_{T_j}=[\mathbf{b}_1,\mathbf{b}_2,\ldots,\mathbf{b}_L]^{T}
\end{equation}
over $E$ such that $\mathbf{b}_i H=\mathbf{c}_i$ for all $i\in [L]$. The matrix $B_{T_j}$ exhibits several properties when $\mathcal{S}$ achieves the cut-set bound with equality, as will be established in the following Lemma \ref{B}.

For a finite field $E$ and the subfield $F$ of $E$, an $m\times n$ matrix $A$ over $E$, denote $\mathcal{S}_F(A)$ the row space of the matrix $A$ over $F$. Specifically, if $A=[a_1,a_2,\ldots,a_m]^{T}$, then
\[\mathcal{S}_F(A)=\left\{k_1a_1+\cdots+k_ma_m\,|\,k_i\in F \text{ for all } i=1,\ldots,m\right\}.\]

\begin{lemma}\label{lem5}\cite[Lemma 5]{tamo2019repair}
	Let $E$ be an extension of a finite field of $K$. Let $A=(a_{i,j})$ be an $m\times n$ matrix over $E$. Then 
	\begin{equation}\label{sfa}
		\dim(\mathcal{S}_K(A))\leq \sum_{j=1}^{n}\dim_K(a_{1,j},a_{2,j},\ldots,a_{m,j}).
	\end{equation}
	Moreover, if \eqref{sfa} holds with equality, then for every $\mathcal{J}\subseteq [n]$,
	\begin{equation}
		\dim(\mathcal{S}_K(A_{\mathcal{J}}))=\sum_{j\in\mathcal{J}}\dim_K(a_{1,j},a_{2,j},\ldots,a_{m,j}),
	\end{equation}
	where $A_{\mathcal{J}}$ is the restriction of $A$ to the columns with indices in $\mathcal{J}$.
\end{lemma}

\begin{lemma}\label{B}
	Let $\mathcal{C}$ be an $(n,k)$ scalar linear MDS code over symbol field $E$ with the base field $\mathbb{F}_q$. Denote the sub-packetization level by $L$, and let $r = n - k$. Suppose $\mathcal{S}$ is a PE repair scheme for $\mathcal{C}$ that achieves the cut-set bound with equality. Let $w$ be the largest integer such that $\sum_{i=1}^w t_i \leq k$. For each $j = 1, 2, \ldots, w$, consider the matrix $B_{T_j}$ defined as in \eqref{btj} with respect to the parity-check matrix $H = [M \mid I_r] = [h_1\ h_2\ \cdots\ h_n]$, where $I_r$ is the identity matrix of order $r$ and $T_j$ is specified in \eqref{ti}. For each $j = 1, 2, \ldots, w$, the matrix $B_{T_j}$ satisfies the following properties:
	\begin{itemize}
		\item[(i)] For $ i \in A_j \setminus \{T_j\} $, $ B_{T_j} h_i = 0 $, where $ 0 $ denotes the zero column vector of length $ L $.  
		\item[(ii)] For $ i \in [n] \setminus A_j $, the dimension of $ B_{T_j} h_i $ over $ \mathbb{F}_q $ is $ L / (r - t_j + 1) $.  
		\item[(iiii)] Each column of $ B_{T_j} $ has dimension $ L / (r - t_j + 1) $ over $ \mathbb{F}_q $.  
		\item[(iv)] There exists an $ L \times L $ invertible matrix $ P $ over $ \mathbb{F}_q $ such that 
		\begin{equation}\label{pbtj}
			P B_{T_j} = [ \mathrm{diag}(p_1, p_2, \dots, p_{r - t_j + 1}) \mid V ],
		\end{equation}
		where each $ p_i $ is a column vector of length $ L / (r - t_j + 1) $ over $ E $ for all $1\leq i\leq r-t_j+1$, and $ V $ is an $ L \times (t_j - 1) $ matrix over $ E $.
	\end{itemize}
\end{lemma} 
\begin{IEEEproof}
	For a fixed $j\in[w]$ and $T_j=\sum_{s=1}^{j-1}t_s+1$, we simplify the notation $B_{T_j}$ to $B$ for clarity and conciseness. We first assert that the $\mathbb{F}_{q}$-rank of the row space of $B$ is $L$. Indeed, suppose to the contrary that there exists a non-zero vector $w \in \mathbb{F}_{q}^{L}$ such that $wB = 0$. Therefore,
	\[
	wBH = w
	\begin{bmatrix}
		c_{1,1} & c_{1,2} & \cdots & c_{1,n} \\
		c_{2,1} & c_{2,2} & \cdots & c_{2,n} \\
		\vdots & \vdots & \vdots & \vdots \\
		c_{L,1} & c_{L,2} & \cdots & c_{L,n}
	\end{bmatrix}
	= 0.
	\]
	This implies that $w(c_{1,T_j}, c_{2,T_j}, \ldots, c_{L,T_j})^T = 0$, which contradicts \eqref{dim}. Thus we conclude that $B$ has $L$ linearly independent rows over $\mathbb{F}_q$.
	For a fixed $j \in [w]$, let $\mathcal{J} \subseteq [n]$ be a subset with $|\mathcal{J}| = r$ such that $A_j \setminus \{T_j\} \subseteq \mathcal{J}$. The existence of $\mathcal{J}$ can be guaranteed since $t_j \leq r$ for all $j \in [l]$. Since $H$ generates an $(n,r)$ MDS code, 
	the submatrix $H_{\mathcal{J}}$ has full rank. Therefore, the $L \times r$ matrix $BH_{\mathcal{J}}$ satisfies the conditions
	\begin{align}\label{LDIM}
		L = \dim(\mathcal{S}_{\mathbb{F}_q}(B)) = \dim(\mathcal{S}_{\mathbb{F}_q}(BH_{\mathcal{J}})) \leq \sum_{i \in \mathcal{J}} \dim_{\mathbb{F}_q}(Bh_i),
	\end{align}
	where the last inequality follows from Lemma \ref{lem5}. Summing both sides of \eqref{LDIM} over all subsets $\mathcal{J}$ such that $A_j\setminus \{T_j\}\subseteq\mathcal{J} \subseteq [n]\setminus \{T_j\}$ of size $|\mathcal{J}| = r$, we obtain that
	\begin{align*}
		L\binom{n - t_j}{r-t_j+1} &\leq \sum_{\substack{A_j\setminus \{T_j\}\subseteq\mathcal{J} \subseteq [n]\setminus \{T_j\} \\ |\mathcal{J}| = r}} \sum_{i \in \mathcal{J}} \dim_{\mathbb{F}_q}(Bh_i) \\
		&= \sum_{\substack{A_j\setminus \{T_j\}\subseteq\mathcal{J} \subseteq [n]\setminus \{T_j\} \\ |\mathcal{J}| = r}}\left(  \sum_{i \in A_j\setminus \{T_j\}} \dim_{\mathbb{F}_q}(Bh_i) +  \sum_{i \in \mathcal{J}\setminus A_j} \dim_{\mathbb{F}_q}(Bh_i)              \right)\\
		&\leq \sum_{\substack{A_j\setminus \{T_j\}\subseteq\mathcal{J} \subseteq [n]\setminus \{T_j\} \\ |\mathcal{J}| = r}}\sum_{i \in \mathcal{J}\setminus A_j} \dim_{\mathbb{F}_q}(Bh_i) \\
		&= \binom{n - t_j-1}{r - t_j} \sum_{i \in [n]\setminus A_j} \dim_{\mathbb{F}_q}(Bh_i)\\
		&\stackrel{(\ref{dim2})}{=} \binom{n - t_j-1}{r - t_j} \frac{(n - t_j)L}{r-t_j+1}\\
		&= L\binom{n - t_j}{r-t_j+1}, 
	\end{align*}
	where the set minus notation $ \mathcal{J}\setminus A_j$ means exclude the elements that both belongs to $\mathcal{J}$ and $A_j$ from $\mathcal{J}$. The inequality above is in fact an equality, therefore
	\begin{equation}\label{28}
		\sum_{i \in A_j\setminus \{T_j\}} \dim_{\mathbb{F}_q}(Bh_i)=0.
	\end{equation}
	Furthermore, for every subset $A_j\setminus \{T_j\}\subseteq\mathcal{J} \subseteq [n]\setminus \{T_j\}$ and $|\mathcal{J}| = r$, we have
	\begin{equation}\label{29}
		L = \sum_{i \in \mathcal{J}\setminus A_j} \dim_{\mathbb{F}_q}(Bh_i).
	\end{equation}
	The conclusion of (i) can be obtained by \eqref{28}.
	Noting that the size of $\mathcal{J}\setminus A_j$ is $r - t_j + 1$, it follows that the conclusion of (ii) can be derived from \eqref{29}, due to the arbitrary choice of $\mathcal{J}$.
	
	We will now proceed to prove conclusion (iii). It follows from (i) that 
	$$L = \dim(\mathcal{S}_{\mathbb{F}_q}(B)) = \dim(\mathcal{S}_{\mathbb{F}_q}(BH_{\mathcal{J}}))=\dim(\mathcal{S}_{\mathbb{F}_q}(BH_{\mathcal{J}\setminus A_j}).$$
	In other words, 
	for every $\mathcal{J} \subseteq [n]\setminus A_j$ and $|\mathcal{J}| = r-t_j+1$, we have
	\begin{equation}
		L = \dim(\mathcal{S}_{\mathbb{F}_q}(BH_{\mathcal{J}})).
	\end{equation}
	According to the second part of Lemma \ref{lem5}, 
	for every $\mathcal{J} \subseteq [n]\setminus A_j$ of size $|\mathcal{J}| \leq r-t_j+1$,
	\begin{equation}\label{31}
		\dim(\mathcal{S}_{\mathbb{F}_q}(BH_{\mathcal{J}})) = \sum_{i \in \mathcal{J}} \dim_{\mathbb{F}_q}(Bh_i) = \frac{|\mathcal{J}|L}{r-t_j+1}.
	\end{equation}
	Let us take $\mathcal{J}$ to be a subset of $\{k+1, k+2, \ldots, n\}$. 
	Since the last $r$ columns of $H$ form an identity matrix, \eqref{31} becomes
	\begin{equation}\label{33}
		\dim(\mathcal{S}_{\mathbb{F}_q}(B_\mathcal{J})) = \frac{|\mathcal{J}|L}{r-t_j+1}
		\quad \text{for all } \mathcal{J} \subseteq [r], \; |\mathcal{J}| \leq r-t_j+1.
	\end{equation}
	The conclusion of (iii) can be deduced by setting $|\mathcal{J}|=1$.
	
	Next we show that by performing elementary row operations over $\mathbb{F}_q$, 
	$B$ can be transformed into a form of \eqref{pbtj}.
	The proof proceeds by induction. More specifically, for $i = 1,2,\ldots,r-t_j+1$, 
	we can use elementary row operations over $\mathbb{F}_q$ to transform the first $i$ columns of $B$ 
	into the following form:
	\[
	\begin{bmatrix}
		p_1 & 0 & \cdots & 0 \\
		0 & p_2 & \cdots & 0 \\
		\vdots & \vdots & \vdots & \vdots \\
		0 & 0 & \cdots & p_i \\
		\mathbf{0} & \mathbf{0} & \cdots & \mathbf{0}
	\end{bmatrix},
	\]
	where each $\mathbf{0}$ in the last row of the above matrix is a column vector of length $L\!\left(1-\tfrac{i}{r-t_j+1}\right)$.
	
	Let $i=1$. According to \eqref{33}, each column of $B$ has dimension $\frac{L}{r-t_j+1}$ over $\mathbb{F}_q$. Thus the induction base holds trivially. Now assume that there is an $L \times L$ invertible matrix $P$ over $\mathbb{F}_q$ such that
	\[
	PB_{[i-1]} =
	\begin{bmatrix}
		p_1 & 0 & \cdots & 0 \\
		0 & p_2 & \cdots & 0 \\
		\vdots & \vdots & \vdots & \vdots \\
		0 & 0 & \cdots & p_{i-1} \\
		\mathbf{0} & \mathbf{0} & \cdots & \mathbf{0}
	\end{bmatrix},
	\]
	where each $\mathbf{0}$ in the last row of this matrix is a column vector of length 
	$L\!\left(1-\tfrac{i-1}{r-t_j+1}\right)$, $B_{[i-1]}$ is a submatrix of $B$ consisting of the first $i-1$ columns of $B$. Let us write the $i$-th column of $PB$ as $(v_1, v_2, \ldots, v_L)^T$. Since each column of $B$ has dimension $\frac{L}{r-t_j+1}$ over $\mathbb{F}_q$, $(v_1, v_2, \ldots, v_L)^T$ also has dimension $\frac{L}{r-t_j+1}$ over $\mathbb{F}_q$. Since the last $L\!\left(1-\tfrac{i-1}{r-t_j+1}\right)$ rows of the matrix $PB_{[i-1]}$ are all zero, we can easily deduce that
	\[
	\dim(\mathcal{S}_{\mathbb{F}_q}(PB_{[i]})) 
	\leq \frac{i-1}{r-t_j+1}L + \dim_{\mathbb{F}_q}(v_{\frac{(i-1)L}{r-t_j+1}+1}, v_{\frac{(i-1)L}{r-t_j+1}+2}, \ldots, v_L).
	\]
	By \eqref{33}, $\dim(\mathcal{S}_{\mathbb{F}_q}(PB_{[i]})) = \dim(\mathcal{S}_{\mathbb{F}_q}(B_{[i]})) = \tfrac{iL}{r-t_j+1}$. As a result,
	\[
	\dim_{\mathbb{F}_q}(v_{(i-1)L/(r-t_j+1)+1}, v_{(i-1)L/(r-t_j+1)+2}, \ldots, v_L) 
	\geq \frac{L}{r-t_j+1} = \dim_{\mathbb{F}_q}(v_1, v_2, \ldots, v_L).
	\]
	In other words, $(v_{(i-1)L/(r-t_j+1)+1}, v_{(i-1)L/(r-t_j+1)+2}, \ldots, v_L)$ contains a basis of the set $(v_1, v_2, \ldots, v_L)$ over $\mathbb{F}_q$. This implies that we can use elementary row operations on the matrix $PB$ to eliminate all the nonzero entries $v_m$ for $m \leq \frac{(i-1)L}{r-t_j+1}$, and thus obtain the desired block-diagonal structure for the first $i$ columns through row exchanges if necessary. This establishes the induction step, thereby establishing the conclusion of (iv).
\end{IEEEproof}
For a vector space $V$ over $\mathbb{F}_q$ and a set of vectors $S=\{\mathbf{v}_1,\ldots,\mathbf{v}_m\}\subseteq V$, we denote the span of $S$ over $\mathbb{F}_q$ as
$$\text{Span}_{\mathbb{F}_{q}}(S)=\left\{\sum_{i=1}^{m}\gamma_i \mathbf{v}_i \,: \,\gamma_i \in \mathbb{F}_{q}\right\}.$$
\begin{lemma}\label{H}
	Let $\mathcal{C}$ be an $(n,k)$ scalar linear MDS code over symbol field $E$ with base field $\mathbb{F}_q$. Denote the sub-packetization level by $L$, and let $r = n - k$. Suppose $\mathcal{S}$ is a linear repair scheme for $\mathcal{C}$ that achieves the cut-set bound with equality. Let $T_1=1$ and $T_i=\sum_{j=1}^{i-1}t_j+1$ for $i=2,\ldots,l$. Let $w$ be the largest integer such that $\sum_{i=1}^w t_i \leq k$. Assume that $w \geq 2$.
	Let $H=[M|I_r]=(h_{i,j})$ be the parity-check matrix of $\mathcal{C}$.
	There exist $\beta_1,\ldots,\beta_{w-1}\in E$ such that
	for every $j=1,2,\ldots,w-1$, 
	\begin{equation}\label{beta}
		\beta_{j}\notin \mathbb{F}_{q}(\beta_{i}:i\in \{1,2,\ldots,w-1\}\backslash\{j\}).
	\end{equation}
\end{lemma}
\begin{IEEEproof}
	For each $j=1,\ldots,w$, according to Lemma \ref{B}, there exists a repair matrix $B_{T_j}$ for node $T_j$ that satisfying the property of (i)-(iv) of Lemma \ref{B}. For a fixed $j\in[w]$, we simplify the notation $B_{T_j}$ to $B$ for clarity and conciseness.  For each $1\leq i\leq n$, let $h_i$ be the $i$-th column of $H$, where $h_i^{(1)}$ be the first $r-t_j+1$ rows of $h_i$ and $h_i^{(2)}$ be the last $t_j-1$ rows of $h_i$.
	Let $P$ be the $L\times L$ invertible matrix over $\fq$ such that 
	\begin{equation}\label{matrix}
		PB=\left[ \begin{matrix}
			\left. \begin{matrix}
				\,\,\, 	{{p}_{1}} & 0 & \cdots  & 0  \\
				0 & {{p}_{2}} & \cdots  & 0  \\
				0 & 0 & \ddots  & \vdots   \\
				0 & 0 & 0 & {{p}_{r-{{t}_{j}}+1}}  \\
			\end{matrix}\,\,\, \right| & V \\
		\end{matrix} \,\,\,  \right],
	\end{equation}
	where $p_i$ is a column vector of length $\frac{L}{r-t_j+1}$ over $E$ and $V$ is an $L\times (t_j-1)$ matrix over $E$.
	According to the property (i) of Lemma \ref{B}, for each $i\in A_j\setminus\{T_j\}$, we have $PBh_{i}=0$. Therefore, 
	\begin{equation}
		Vh_i^{(2)}=\begin{bmatrix}
			-h_{1,i}p_1 \\
			-h_{2,i} p_2\\
			\vdots \\
			-h_{r-t_j+1,i} p_{r-t_j+1}\\
		\end{bmatrix}.
	\end{equation}
	Denote $A_j^{*}=A_j\setminus\{T_j\}=\{e_1,\ldots,e_{t_j-1}\}$. Since $H$ is a generator matrix of an MDS code and $t_j-1<r$ for all $j\in[l]$, every $(t_j - 1) \times (t_j - 1)$ submatrix of $H$ is invertible.
	Therefore, for any $i\in [n]\setminus A_j^{*}$, we have $h_i^{(2)}\in \text{Span}_E(\{h_{e_1}^{(2)},\ldots,h_{e_{t_j-1}}^{(2)}\})$.
	Consequently, for any $i\in [n]$, let 
	\begin{equation}\label{37}
		h_i^{(2)}=\sum_{m=1}^{t_j-1}a_{m,i}h_{e_m}^{(2)},
	\end{equation}
	where $a_{m,i}\in E$ for all $m\in [t_j-1]$. Therefore, for each $i\in [n]\setminus A_j$,
	\begin{equation}
		Vh_i^{(2)}=\sum_{m=1}^{t_j-1}a_{m,i}Vh_{e_m}^{(2)}=\begin{bmatrix}
			-\sum_{m=1}^{t_j-1}a_{m,i}h_{1,e_m} p_1\\
			-\sum_{m=1}^{t_j-1}a_{m,i}h_{2,e_m} p_2\\
			\vdots \\
			-\sum_{m=1}^{t_j-1}a_{m,i}h_{r-t_j+1,e_m} p_{r-t_j+1}\\
		\end{bmatrix}.
	\end{equation}
	Hence, for each $i\in [n]\setminus A_j$, we have 
	\begin{equation}\label{39}
		PBh_i=\begin{bmatrix}
			p_1(h_{1,i}-\sum_{m=1}^{t_j-1}a_{m,i}h_{1,e_m})\\
			p_2(h_{2,i}-\sum_{m=1}^{t_j-1}a_{m,i}h_{2,e_m})\\
			\vdots \\
			p_{r-t_j+1}(h_{{r-t_j+1},i}-\sum_{m=1}^{t_j-1}a_{m,i}h_{{r-t_j+1},e_m})\\
		\end{bmatrix}.
	\end{equation}
	For each $u\in[r-t_j+1]$ and $i\in [n]\setminus A_j$, let $\alpha_{u,i}=h_{u,i}-\sum_{m=1}^{t_j-1}a_{m,i}h_{u,e_m}$.
	Note that these elements $\alpha_{u,i}$ are independent of the choice of matrix $B$, and thus remain identical when considering the repair of any nodes $T_j$ for $j\in [w]$.
	We claim that for each $u\in[r-t_j+1]$ and for each $i\in [n]\setminus A_j^{*}$, $\alpha_{u,i}$ is a nonzero element in $E$. 
	Since $t_j\leq r$ for all $j\in [l]$, it follows that any $t_j$ rows of the matrix $[h_i, h_{e_1}, \ldots, h_{e_{t_j-1}}]$ are $E$-linearly independent. Suppose, for contradiction, that $\alpha_{u,i} = 0$ for some $u \in [r - t_j + 1]$ and some $i\in [n]\setminus A_j^{*}$. Then, $h_{u,i}=\sum_{m=1}^{t_j-1}a_{m,i}h_{u,e_m}$ for such $u \in [r - t_j + 1]$ and $i\in [n]\setminus A_j^{*}$.
	Combining with \eqref{37}, it follows that the vectors formed by selecting the $u$-th row and the last $t_j - 1$ rows from the matrix $[h_i, h_{e_1}, \ldots, h_{e_{t_j-1}}]$ are linearly dependent over $E$, contradicting their linear independence. Hence, $\alpha_{u,i} \ne 0$ for all $u \in [r - t_j + 1]$ and $i\in [n]\setminus A_j^{*}$.
	For $u \in [r-t_j+1]$, let $P_u$ be the vector space spanned by the entries of $p_u$ over $\mathbb{F}_q$. According to (ii) of Lemma \ref{B}, for all $i \in [n]\setminus A_j$,  
	\[
	\dim_{\mathbb{F}_q}(PBh_i) = \dim_{\mathbb{F}_q}(Bh_i) = \frac{L}{r-t_j+1}.
	\]
	According to \eqref{39}, we have
	\[
	\begin{aligned}
		\dim_{\mathbb{F}_q}(PBh_i) =\dim_{\mathbb{F}_q}(P_1 \alpha_{1,i} + \cdots + P_{r-t_j+1} \alpha_{r-t_j+1,i}), i \in [n]\setminus A_j.
	\end{aligned}
	\]
	It follows that for all $i \in [n]\setminus A_j$,
	\begin{equation}\label{24}
		\dim_{\mathbb{F}_q}(P_1 \alpha_{1,i} + \cdots + P_{r-t_j+1} \alpha_{r-t_j+1,i}) = \frac{L}{r-t_j+1}.
	\end{equation}
	Since each column of $B$ has dimension $\frac{L}{r-t_j+1}$ over $\mathbb{F}_q$, $P_u$ also has dimension $ \frac{L}{r-t_j+1}$ over $\mathbb{F}_q$ for every $u \in [r-t_j+1]$. Recall that $\alpha_{u,i} \neq 0$ for all $u \in [r-t_j+1]$ and all $i \in  [n]\setminus A_j^*$. Thus
	\begin{equation}\label{dim3}
		\dim_{\mathbb{F}_q}(P_u \alpha_{u,i}) = \frac{L}{r-t_j+1}
	\end{equation}
	for all $u = 1, \ldots, r-t_j+1$ and $i \in  [n]\setminus A_j^*$. Therefore, combining \eqref{24} and \eqref{dim3}, we have
	\begin{equation}\label{p11}
		P_1 \alpha_{1,i} = P_2 \alpha_{2,i} = \cdots = P_{r-t_j+1} \alpha_{r-t_j+1,i},\text{ for all } i \in  [n]\setminus A_j.
	\end{equation}
	Since $\alpha_{1,i}$ is a nonzero element determined solely by the parity-check matrix $H$ for each $i \in [n]\setminus A_j^{*}$, without loss of generality, we may assume that $\alpha_{1,i} = 1$ for all $i \in [k]\setminus A_j^{*}$ due to the code equivalence.
%
%
%
Therefore, we obtain
	\begin{equation}\label{p12}
		P_1=P_1\alpha_{1,i_1}=P_1\alpha_{1,i_2}, \text{ for any } i_1,i_2\in [k]\setminus A_j^{*}.
	\end{equation}
	Combining \eqref{p11} and \eqref{p12}, we obtain that
	\begin{equation}
		P_u\alpha_{u,i_1}=P_u\alpha_{u,i_2}, \text{ for all } u \in [r-t_j+1] \text{ and }i_1,i_2\in [k]\setminus A_j.
	\end{equation}
	Let $\beta_i=\frac{\alpha_{2,T_i}}{\alpha_{2,T_w}}$, where $i=1,2,\ldots,w-1$.
	Therefore, for all $i \in \{1,2,\ldots, w-1\}\setminus\{j\}$,
	\[P_2 \beta_i = P_2.\]
	By definition $P_2$ is a vector space over $\mathbb{F}_q$, so
	\begin{equation}\label{p22}
		P_2 \gamma =P_2 \quad \text{for all } \gamma \in \mathbb{F}_q(\{\beta_i : i \in \{1,2,\ldots, w-1\}\setminus \{j\}\}). 
	\end{equation}
	On the other hand,
	\begin{align}\label{43}
		\dim_{\mathbb{F}_q}(P_1 \alpha_{1,T_j} + \cdots + P_{r-t_j+1} \alpha_{r-t_j+1,T_j})= \dim_{\mathbb{F}_q}\{c_{1,T_j}, c_{2,T_j}, \ldots, c_{l,T_j}\} = L,
	\end{align}
	while according to \eqref{dim3} 
	\begin{equation}\label{44}
		\dim_{\mathbb{F}_q}(P_u \alpha_{u,T_j}) =  \frac{L}{r-t_j+1} , \quad u = 1,2,\ldots,r-t_j+1.
	\end{equation}
	It follows from \eqref{43} and \eqref{44} that the vector spaces $P_1 \alpha_{1,T_j}, P_2 \alpha_{2,T_j}, \ldots, P_{r-t_j+1} \alpha_{r-t_j+1,T_j}$ are pairwise disjoint. According to \eqref{p11} and \eqref{p12}, we have $P_u \alpha_{u,i}=P_1 \alpha_{1,T_j}=P_1$ for all $u \in [r-t_j+1]$ and $i \in [k] \setminus A_j$, it follows that
	\[
	P_u \alpha_{u,i} \cap P_2 \alpha_{2,T_j} = \{0\}, \quad \text{for all } u \in [r-t_j+1] \text{ and } i \in [k] \setminus A_j.
	\]
	Setting $u = 2$, it follows that $P_2 \alpha_{2,T_i} \cap P_2 \alpha_{2,T_j} = \{0\}$ for all $i \in \{1,2,\ldots, w\} \setminus \{j\}$.
	Taking $i = w$, we deduce $P_2 \beta_j \neq P_2$. Then, by \eqref{p22}, we conclude that $\beta_j \notin \mathbb{F}_q(\{\beta_i : i \in \{1,2,\ldots, w-1\} \setminus \{j\}\})$.
\end{IEEEproof}

\begin{theorem}\label{our_bound}
	Let $F=\mathbb{F}_{q}$ and $E=\mathbb{F}_{q^{L}}$ for a prime power $q$. Let $\mathcal{C} \subseteq E^{n}$ be an $(n, k)$ scalar linear MDS code over the base field $F$. 
	Suppose $\mathcal{S}$ is a PE linear repair scheme for $\mathcal{C}$ that achieves the cut-set bound with equality. 
	Let the exclusion sets for $\mathcal{S}$ be $A_{1}, A_{2},\ldots,A_{l}$, where each set $A_{i}$ has cardinality $t_{i}$ for $1\leq t_i\leq \min\{k,n-k\}$ and $i \in [l]$. Assume that $t_1\leq \cdots\leq t_l$. Let $w$ be the greatest integer such that $\sum_{i=1}^{w}t_i\leq k$.
	Then the sub-packetization level $L$ is at least 
	\begin{equation}\label{sub-packbound}
		L	\ge \prod\limits_{i=1}^{w-1}p_i,
	\end{equation}
	where $p_{i}$ is the $i$-th smallest prime, and the product equals $1$ when $w = 1$.
\end{theorem}
\begin{IEEEproof}
	For $w=1$, the inequality is always satisfied. 
	Therefore, it suffices to consider the case where $w \geq 2$. 
	Let $m=w-1$. 
	According to Lemma \ref{H}, 
	there exist $\beta_1,\ldots,\beta_{m}\in E$ such that
	for each $j=1,2,\ldots,m$, 
	\begin{equation}\label{beta1}
		\beta_{j}\notin \mathbb{F}_{q}(\beta_{i}:i\in [m]\backslash\{j\}).
	\end{equation}
	For $1\leq i\leq m$, let $d_{i}$ be the degree of $\beta_{i}$ over $\fq$. 
	Condition \eqref{beta1} immediately yields $d_i > 1$ for all $i$.
	The extension degree of $\fq(\beta_1,\beta_2,\ldots,\beta_m)$ over $\fq$ is $d=\mathrm{lcm}(d_1,d_2,\ldots,d_m).$
	Correspondingly, for $1 \le j \le m$, we have 
	\[[\mathbb{F}_q(\beta_i : i\in [m]\backslash\{j\}): \mathbb{F}_q] = \bar{d}_j = \operatorname{lcm}(d_i : i\in [m]\backslash\{j\}).\]
	 The condition \eqref{beta1} implies that for each $1\leq j\leq m$, 
	$d_j \nmid \bar{d}_j$.
	By considering the prime factorization of $d_j$, it follows that for each $1\leq j\leq m$, there exists a prime number $p_j$ that divides $d_j$ and is coprime to $\bar{d}_j$.
	We assert that the primes $\{p_j\,:\, 1\leq j\leq m\}$ are pairwise distinct. Suppose, to the contrary,
	that $p_i=p_j$ for some $i\neq j$. Observe that $p_j\, |\, d_j$ and $d_j \, |\,\bar{d}_i$, hence $p_j\,|\,\bar{d}_i$.
	By the assumed equality, $p_i\,|\,\bar{d}_i$ follows. This, however, contradicts the given fact that $\text{gcd}(p_i,\bar{d}_i)=1$.
	Therefore, the primes are distinct.
    Since each prime divides $d$, we have $d \geq \prod_{i=1}^{m} p_i$. 
    Since $\fq(\beta_1,\beta_2,\ldots,\beta_m)$ is a subfield of $E$, it follows that $L\geq d\geq \prod_{i=1}^{m} p_i$.
    Furthermore, the lower bound continues to hold when each $p_i$ is taken as the $i$-th smallest prime.
\end{IEEEproof}

In the case where exclusion sets with a uniform size, we can directly derive the following lower bound from Theorem~\ref{our_bound}.
\begin{theorem}\label{cor_bound}
	Let $F=\mathbb{F}_{q}$ and $E=\mathbb{F}_{q^{L}}$ for a prime power $q$. Let $\mathcal{C} \subseteq E^{n}$ be an $(n, k)$ scalar linear MDS code over the base field $F$. 
	Suppose $\mathcal{S}$ is a PE repair scheme for $\mathcal{C}$ that achieves the cut-set bound with equality.
	Let the exclusion sets for $\mathcal{S}$ be $A_{1}, A_{2},\ldots,A_{l}$, where all sets have an equal cardinality $t$ for $1\leq t\leq \min\{k,n - k\}$.
	Then the sub-packetization level $L$ is at least 
	\begin{equation}\label{sub-packbound_1}
		L	\ge \prod\limits_{i=1}^{\left\lfloor k/t \right\rfloor-1}p_i,
	\end{equation}
	where $p_{i}$ is the $i$-th smallest prime, and the product equals $1$ when $\left\lfloor k/t \right\rfloor = 1$.
\end{theorem}

By the Prime Number Theorem \cite{iwaniec2021analytic}, for the case where all exclusion sets have a given equal size~$t$, the lower asymptotic bound on $L$ is $e^{(1+o(1))\lfloor k/t \rfloor \log(\lfloor k/t \rfloor)}$. Furthermore, if we set $t$ to be a quantity related to $k$, then the lower bound of $L$ turns into a constant. For instance, if $3t\leq k<4t$, then the lower bound becomes $L\geq 6$.

\section{A Trade-Off for PE Repair Schemes}\label{tradeoffsec}
For clarity, in this section, we assume a uniform value $t_i = t$ for all $i\in[n]$.
Consider an $(n,k)$ scalar MDS code under the PE repair framework with flexibility $t$. 
The \textit{smallest possible sub-packetization level} refers to a lower bound of the sub-packetization level for which MDS codes can achieve the cut-set bound for single-node repair.
According to Theorem~\ref{cor_bound}, an increase in the flexibility $t$ leads to a decrease in the smallest possible sub-packetization level $L$.
However, increasing $t$ also gives rise to a higher minimum repair bandwidth.
This is because a larger $t$ also causes a decrease in the maximum allowable repair locality $d_{max}$, which in turn increases the minimum repair bandwidth.
More specifically, the cut-set bound indicates that the minimum possible repair bandwidth to repair a failed node with accessing $d$ helper nodes is $\frac{dL}{d - k+1}$. This value decreases as $d$ increases. 
Therefore, the repair bandwidth can be minimized in a two-layer manner when it reaches the cut-set bound under the condition that $d$ attains its maximum value. 

To precisely characterize how $t$ influences the repair performance of scalar MDS codes, one should examine both the smallest possible sub-packetization level $L$ and the minimum repair bandwidth jointly. Since the value of $L$ influences the repair bandwidth, a fair comparison requires normalizing the repair bandwidth relative to $L$. We thus introduce the concept of \textit{normalized repair bandwidth}. For a repair scheme with repair bandwidth $\beta$ and sub-packetization level $L$, the normalized repair bandwidth of the repair scheme is given by
$\bar{\beta}=\frac{\beta}{L}$, which represents the number of bits transmitted per repaired bit. For a given flexibility $t$, the \textit{minimum normalized repair bandwidth} equals $\bar{\beta}_{min}=\frac{n-t}{n-t-k+1}$. 

We have demonstrated the trade-off curve for practical $(14,10)$~RS codes over the base field $\mathbb{F}_2$ in Fig.~1.
In the scenario of $(14,10)$ RS codes, the range of allowable $t$ is narrow, obscuring the trend of $L$ and $\bar{\beta}_{min}$ at intermediate points (corresponds to new repair modes). To capture that progression,  we therefore also plot the trade-off curve for the $(20,10)$ scalar MDS codes in Fig. 2, where the evaluation of $t$ ranges from 1 to 10.

\begin{figure}[ht]%
	\centering
	\label{tradeoff_curve_for_20}
	\includegraphics[width=0.65\textwidth]{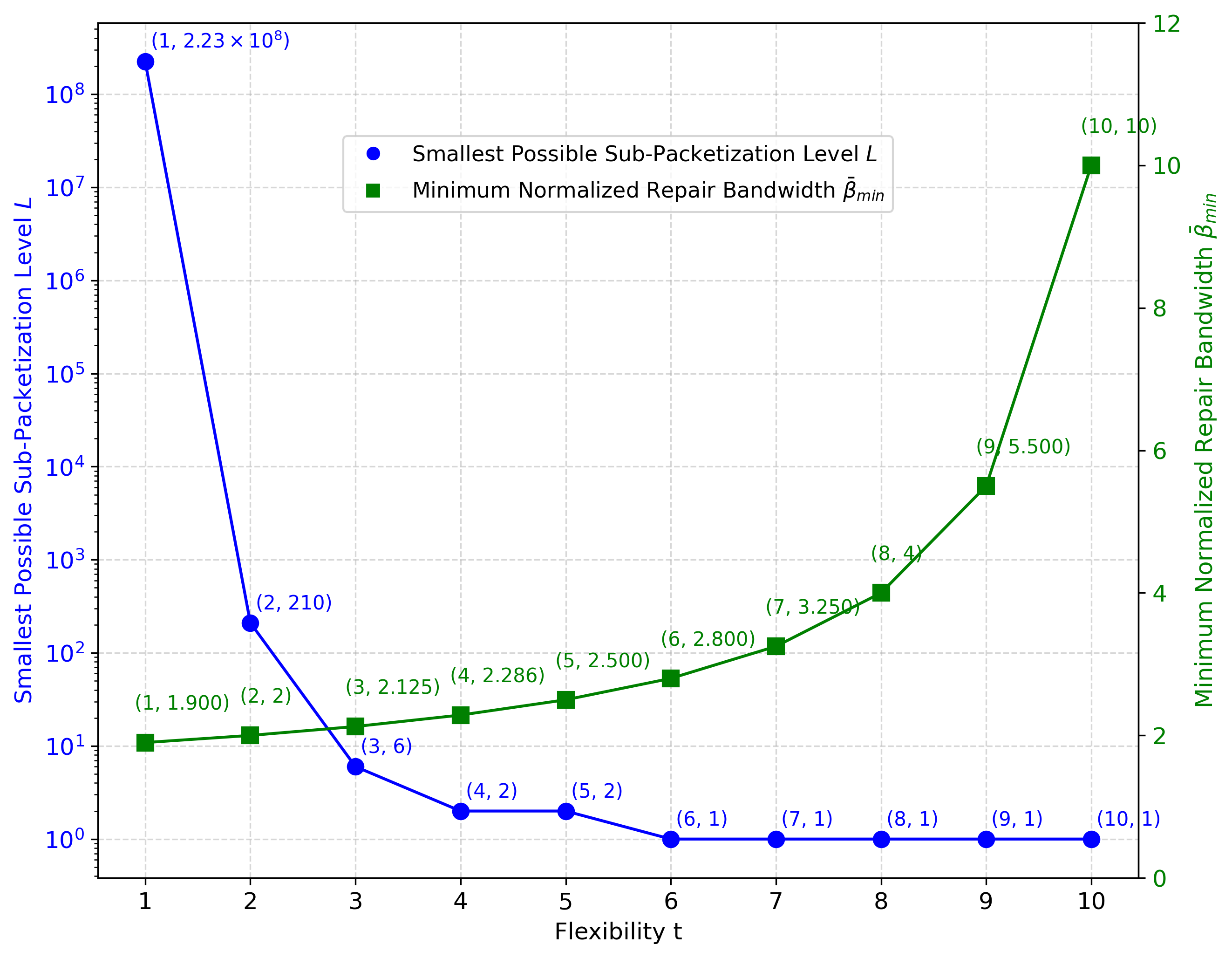}
	\caption{The trade-off among the flexibility $t$, the smallest possible sub-packetization level $L$ and the minimum normalized repair bandwidth $\bar{\beta}_{min}$, for $(20,10)$ RS codes over a base field $\mathbb{F}_q$.}
\end{figure}

Fig. 2 indicates that when sub-packetization is the primary consideration, $t = 6$ is a more favorable option. 
When the repair bandwidth is prioritized, $t = 3$ or 4 becomes a more preferable choice.
Consequently, in a specific practical scenario, one may identify the appropriate $t$ by plotting its particular trade-off curve.

Although the lower bound of the PE repair scheme with  flexibility $t\geq 2$ is strictly smaller than that of the conventional repair framework, the practical advantage of the PE repair framework remains elusive because it is unclear whether the lower bound is achievable.
To explicitly demonstrate that the PE repair scheme can indeed yield codes with better repair performance than the conventional one, in the subsequent two sections, we present two generic constructions of RS codes that achieve the cut-set bound under the PE repair scheme, with sub-packetization levels significantly less than the lower bound of sub-packetization level established for the conventional repair framework.

\section{The First Construction of RS Codes with Optimal Repair Bandwidth}\label{construction1}
In this section, we construct a class of RS codes that admit a nontrivial optimal repair scheme achieving the cut-set bound~\eqref{cut-set-bound} for any single failed node.
We begin by specifying the symbol field and the base field of the constructed RS codes, as well as the associated finite fields that will be utilized in the repair scheme.
\subsection{The symbol field and the base field}
Let $\mathbb{F}_q$ be a finite field, where $q$ is a power of a prime.
Let $s\geq 2$ be a positive integer, the value of which is determined by the dimension of the constructed RS code and the repair locality of the equipped repair scheme, as will be specified subsequently.
Let $l\geq 2$ be a positive integer and $p_1,p_2,\ldots,p_l$ be $l$ distinct primes such that
\begin{equation}\label{p}
p_i\equiv 1 \bmod s
\end{equation}
for all $1\leq i\leq l$.
There are infinitely many such primes according to Dirichlet's prime number theorem. For a fixed $i \in [l]$, let $\mathcal{R}_i$ be a nonempty subset of $[l] \setminus \{i\}$. Define three notations as follows:
\begin{equation}\label{u}
\bar{u}_i = \prod_{j \in \mathcal{R}_i} p_j, \quad u_i = \prod_{j \in [l] \setminus \{i\}} p_j, \quad u = \prod_{j \in [l]} p_j.
\end{equation}
It is straightforward to verify that $\bar{u}_i$ is a divisor of $u_i$, and $u_i$ is a divisor of $u$. According to subfield criterion theorem \cite[Theorem 2.6]{lidl97}, it can be deduced that for each $i \in [l]$, the finite field $\mathbb{F}_{q^{p_i}}$ is an extension field of $\mathbb{F}_q$. Moreover, the intersection of any two distinct extension fields $\mathbb{F}_{q^{p_i}}$ and $\mathbb{F}_{q^{p_j}}$ is precisely the finite field $\mathbb{F}_q$, as $\gcd(p_i, p_j) = 1$ for all distinct indices $i,j\in [l]$.
Additionally, for each $i \in [l]$, the following inclusion relationship holds:
\begin{equation}\label{field_notation}
\mathbb{F}_q \subseteq \mathbb{F}_{q^{\bar{u}_i}} \subseteq \mathbb{F}_{q^{u_i}} \subseteq \mathbb{F}_{q^u} \subseteq \mathbb{F}_{q^{us}}.
\end{equation}

We set that the symbol field of constructed RS codes is $\mathbb{F}_{q^{us}}$ and the base field of constructed RS codes is $\mathbb{F}_{q}$.

\begin{example}\label{exam1}
Set the base field as $\mathbb{F}_2$ and $l=4$. Assuming that $s=2$. The smallest four prime numbers that satisfy condition \eqref{p} are $3,\,5,\,7$, and $11$.
Thus, let the symbol field be $\mathbb{F}_{2^{2310}}$. All the subfields of the finite field $\mathbb{F}_{2^{1155}}$ can be determined by listing all positive divisors of $1155$, which has been shown in Figure~\ref{subfields}.
\begin{figure}[h]%
\centering
\label{subfields}
\caption{Inclusion Relationships Among Associated Subfields of the Symbol Field $\mathbb{F}_{2^{2310}}$ }
\includegraphics[width=0.4\textwidth]{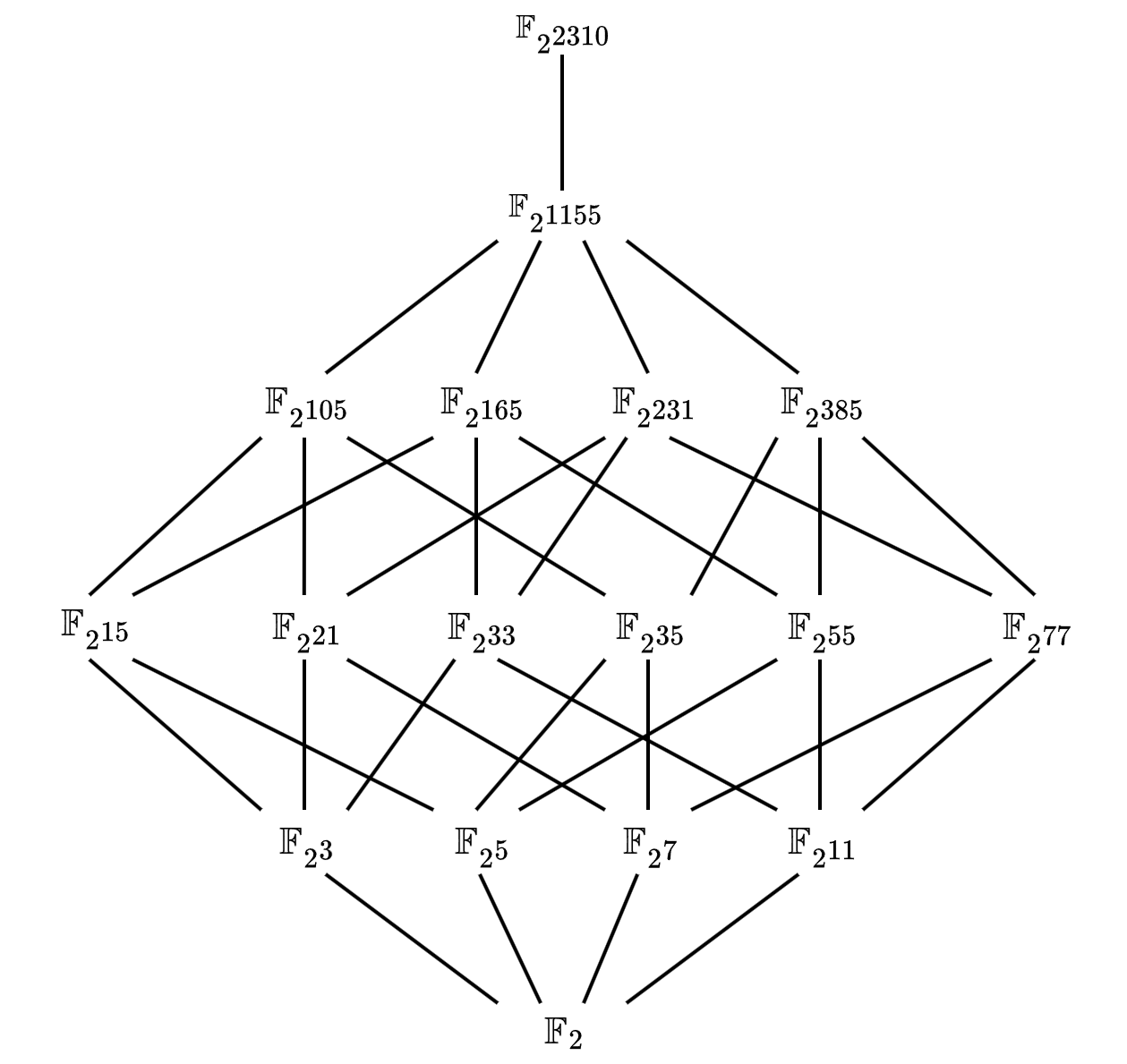}
\end{figure}
\end{example}
\subsection{Construction of RS codes}
Let $L=us$ be the sub-packetization level of RS codes to be constructed. To construct the RS codes, we must specify their evaluation points. 
We design these points to be elements in a finite field that possess both the properties of being a primitive element and a defining element.
Recall that a \textit{primitive element} of a finite field $\mathbb{F}_q$ is a generator of its multiplicative group.
For any subset $M$ of $\mathbb{F}_{q^L}$, the field $\fq(M)$ is defined as the intersection of all subfields of $\mathbb{F}_{q^L}$
containing both $\fq$ and $M$ and is called the extension field of $\fq$ obtained by adjoining the elements in $M$. If $M$ consists of a single element $\alpha\in \mathbb{F}_{q^L}$
and $\fq(\alpha)$ is exactly $\mathbb{F}_{q^L}$, then $\alpha$ is called a \textit{defining element} of $\mathbb{F}_{q^L}$ over $\fq$.
By \cite[Theorem 2.10]{lidl97}, for any positive integer $p$, any primitive element of a finite extension field $\mathbb{F}_{q^p}$ is also a defining element of $\mathbb{F}_{q^p}$ over $\mathbb{F}_q$.

For each $1 \leq i \leq l$, we select $t_i$ distinct primitive elements from $\mathbb{F}_{q^{p_i}}$. Note that the number of primitive elements in a finite field $\fq$ is given by $\phi(q-1)$, where $\phi(x)$ denotes Euler's totient function, which counts the number of positive integers up to $x$ that are coprime to $x$. Consequently, for each $i \in [l]$, the following inequality must be satisfied: $$1\leq t_i \leq \phi(q^{p_i}-1).$$
For each $i \in [l]$, let $\alpha_{i,1}, \alpha_{i,2}, \ldots, \alpha_{i,t_i}$ denote $t_i$ distinct primitive elements of $\mathbb{F}_{q^{p_i}}$. For each $1\leq i\leq l$, let
\[A_i = \{\alpha_{i,j_i}\, :\, 1 \leq j_i \leq t_i\}.\]

For each $1\leq i\leq l$ and $1 \leq j_i \leq t_i$, since $\alpha_{i,j_i}$ is a primitive element of $\mathbb{F}_{q^{p_i}}$, all other elements in $A_i$ lie within $\fq(\alpha_{i,j_i})$.
Consequently, $\fq(\alpha_{i,j_i})=\fq(A_i)$. By \cite[Theorem 2.10]{lidl97}, the primitive element $\alpha_{i,j_i}$ is also a defining element of $\mathbb{F}_{q^{p_i}}$ over $\fq$. 
Hence, $$\mathbb{F}_{q^{p_i}}=\fq(\alpha_{i,j_i})=\fq(A_i).$$
For any distinct $i_1$, $i_2\in [l]$, we claim that 
\begin{equation}\label{two}
	\mathbb{F}_{q^{p_{i_1}p_{i_2}}}=\fq(A_{i_1},A_{i_2})=\mathbb{F}_{q^{p_{i_1}}}(A_{i_2})=\mathbb{F}_{q^{p_{i_1}}}(\alpha_{i_2,j_{i_2}})
\end{equation}
holds for any $1\leq j_{i_2}\leq t_{i_2}$. 
The equality $\mathbb{F}_{q^{p_{i_1}}}(A_{i_2})=\mathbb{F}_{q^{p_{i_1}}}(\alpha_{i_2,j_{i_2}})$ follows because $\alpha_{i_2,j_{i_2}}$ is a primitive element of $\mathbb{F}_{q^{p_{i_2}}}$ and $A_{i_2}\subseteq \mathbb{F}_{q^{p_{i_2}}}$.
By the subfield criterion theorem \cite[Theorem 2.6]{lidl97}, we deduce that $\mathbb{F}_{q^{p_{i_1}p_{i_2}}}=\fq(A_{i_1},A_{i_2})$. 
Since $[\mathbb{F}_{q^{p_{i_1}p_{i_2}}}:\mathbb{F}_{q^{p_{i_1}}}]=p_{i_2}$ and $A_{i_2}$ is contained in $\mathbb{F}_{q^{p_{i_1}p_{i_2}}}$,
therefore the extension degree of $\mathbb{F}_{q^{p_{i_1}}}(A_{i_2})$ over $\mathbb{F}_{q^{p_{i_1}}}$ is a divisor of $p_{i_2}$.
Combining the fact that $p_{i_2}$ is prime and $A_{i_2}$ is not contained in $\mathbb{F}_{q^{p_{i_1}}}$, we have $[\mathbb{F}_{q^{p_{i_1}}}(A_{i_2}):\mathbb{F}_{q^{p_{i_1}}}]=p_{i_2}$,
yielding $\mathbb{F}_{q^{p_{i_1}p_{i_2}}}=\mathbb{F}_{q^{p_{i_1}}}(A_{i_2})$. This argument \eqref{two} can generalize inductively to any set of distinct primes. 
Therefore,
for each $i\in[l]$ and $1\leq j_i\leq t_i$, we have $\alpha_{i,j_i}$ is a defining element of $\mathbb{F}_{q^{u}}$ over $\mathbb{F}_{q^{u_i}}$, \textit{i.e.},
\begin{equation}\label{any}
	\mathbb{F}_{q^{u}}=\mathbb{F}_{q^{u_i}}(\alpha_{i,j_i}).
\end{equation}

\noindent\textbf{Construction 1.} Let $ l\geq 2 $ be a positive integer, and let $ t_1, t_2, \ldots, t_l $ denote $l$ positive integers. Let $n=\sum_{i=1}^{l}t_i$ and $t=\max_{i\in[l]}t_i$.
Let $k$ be a positive integer satisfying $1\leq k\leq n-t$. Let $d$ be a positive integer such that $k\leq d\leq n-t$. 
Denote $s=d-k+1$. Let $p_1,p_2,\ldots,p_l$ be $l$ distinct primes such that $p_i\equiv 1 \bmod s$ and $ t_i\leq \phi(q^{p_i}-1)$ for all $1\leq i\leq l$. For each $i \in [l]$, let $\alpha_{i,1}, \alpha_{i,2}, \ldots, \alpha_{i,t_i}$ denote $t_i$ distinct primitive elements of the finite field $\mathbb{F}_{q^{p_i}}$, let $A_i = \{\alpha_{i,1}, \alpha_{i,2}, \ldots, \alpha_{i,t_i}\}$ be the $i$-th exclusion set.
Denote the symbol field $\mathbb{F}_{q^{us}}$ as $E$, where $u$ is defined as in \eqref{u}. Let $A = \{\alpha_{i,j_i} : 1 \leq i \leq l, 1 \leq j_i \leq t_i\}$. Define the Reed-Solomon code $RS(n,k,A)_{E}$ as
\[\left\{(f(\alpha_{i,j_i}),1 \leq i \leq l, 1 \leq j_i \leq t_i)\, : \, f(x)\in E[x],\, \deg(f)<k \right\}.\]

\begin{remark}
The code length of Construction~1 is $n=\sum_{i=1}^{l}t_i$, where $ 1\leq t_i\leq \phi(q^{p_i}-1)$ for each $1\leq i\leq l$.
We remark that this construction affords considerable flexibility in the choice of code length over a fixed symbol field $\mathbb{F}_{q^L}$, because each $t_i$ may be selected from a wide range. To illustrate, even when $q$ takes its minimum value $2$, the bound $\phi(2^{p}-1)$ can be large. The number $2^{p}-1$ is a Mersenne number. If $p=5$, $2^{5}-1$ is a prime $31$ and $\phi(2^{5}-1)$ is equal to $30$.
If $p=7$, $2^{7}-1$ is a prime $127$ and $\phi(2^{7}-1)$ is equal to $126$. For larger primes $p$, the value of $\phi(q^{p_i}-1)$ grows dramatically, making the range of admissible $t_i$ (and thus of $n$) highly variable. 
\end{remark}
\begin{example}\label{exam2}
	As demonstrated in Example~\ref{exam1}, consider the base field $\mathbb{F}_2$, $s=2$ and $l=4$, and let the symbol field be $\mathbb{F}_{2^{2310}}$.  The parameters $t_1,\ldots,t_4$ satisfy $1\leq t_1\leq 6$, $1\leq t_2\leq 30$, $1\leq t_3\leq 126$ and $1\leq t_4\leq 1936$. Construction 1 yields RS codes whose lengths $n = \sum_{i=1}^4 t_i$ can be flexibly chosen within the range $4 \leq n \leq 2098$.
\end{example}

\subsection{Repair Scheme for a Single Erasure}\label{single}
Unless otherwise specified, the notations $\bar{u}_i$, $u_i$, and $u$ are defined as in \eqref{u}.

\begin{lemma}\cite[Lemma 1]{tamo2019repair}\label{1}
Let $\beta$ be a generator element of $\mathbb{F}_{q^{us}}$ over $\mathbb{F}_{q^{u}}$. Let $\alpha_i$ be a defining element of $\mathbb{F}_{q^{u}}$ over the field $\mathbb{F}_{q^{u_i}}$ for each $i\in [l]$. Given $i\in [l]$, define the following vector spaces over $\mathbb{F}_{q^{u_i}}$:
\begin{align*}
 & S_i^{(1)}=\text{Span}_{\mathbb{F}_{q^{u_i}}}(\beta^{\mu}\alpha_i^{\mu+\varsigma s},\mu=0,1,\ldots,s-1;
\,\varsigma=0,1,\ldots,\frac{p_i-1}{s}-1),\\
&S_i^{(2)} =\text{Span}_{\mathbb{F}_{q^{u_i}}}(\sum\limits_{t=0}^{s-1}\beta^t\alpha_i^{p_i-1}),\\
&S_i=S_i^{(1)}+S_i^{(2)}.
\end{align*}
Then
\begin{equation*}
\dim_{\mathbb{F}_{q^{u_i}}}S_i=p_i,\quad S_{i}+\alpha_{i}S_{i}+\cdots+\alpha_{i}^{s-1}S_{i}=\mathbb{F}_{q^{us}}.
\end{equation*}
\end{lemma}
\begin{lemma}\cite[Theorem 1.86]{lidl97}\label{5}
	Let $\theta \in F$ be algebraic of degree $n$ over $K$ and let $g$ be the minimal polynomial of $\theta$ over K. Then:
	\begin{enumerate}
		\item[(i)] $K(\theta)$ is isomorphic to $K[x]/(g)$.
		\item[(ii)] $[K(\theta):K]=n$ and $\{1,\theta,\ldots,\theta^{n-1}\}$ is a basis of $K(\theta)$ over $K$.
		\item[(iii)] Every $\alpha \in K(\theta)$ is algebraic over $K$ and its degree over $K$ is a divisor of $n$.
	\end{enumerate}
\end{lemma}
\begin{theorem}\label{vector_space}
Let $t_i$ and $\alpha_{i,j_i}$ be as defined in Construction 1 for each $i\in[l]$ and $j_i\in[ t_i]$.
For each $i\in[l]$, let $\mathcal{R}_{i}$ be a subset of $[l]\setminus\{i\}$.
The notations $\bar{u}_i$, $u_i$ and $u$ are defined as in \eqref{u}.	
Let $e$ be a positive integer such that $1\leq e< q^{p_i}-1$ and $\frac{q^{p_i}-1}{q-1}\nmid e$ for all $i\in[l]$.
Then there exists a subspace $S_{\mathcal{R}_{i}}$ of $\mathbb{F}_{q^{us}}$ over $\mathbb{F}_{q^{\bar{u}_i}}$ such that for every primitive element $\alpha_{i,j_i}$ in $\mathbb{F}_{q^{p_i}}$ with $1\leq j_i\leq t_i$, the following holds:
\begin{equation}\label{sr1}
	\dim_{\mathbb{F}_{q^{\bar{u}_i}}}S_{\mathcal{R}_{i}}=\frac{u}{\bar{u}_i},\quad S_{\mathcal{R}_{i}}+\alpha_{i,j_i}^eS_{\mathcal{R}_{i}}+\cdots+\alpha_{i,j_i}^{e(s-1)}S_{\mathcal{R}_{i}}=\mathbb{F}_{q^{us}}.
\end{equation}
\end{theorem}
\begin{IEEEproof}
The simple algebraic extension field $\mathbb{F}_{q^{u_i}}(\alpha_{i,j_i})$ satisfies that $\alpha_{i,j_i}^e \in \mathbb{F}_{q^{u_i}}(\alpha{i,j_i})$ for any positive integer $e$. 
Since \eqref{any} establishes that $\mathbb{F}_{q^{u}} = \mathbb{F}_{q^{u_i}}(\alpha_{i,j_i})$ for each $i$ and $j_i$, it follows from Lemma~\ref{5} that the degree of $\alpha_{i,j_i}^e$ over $\mathbb{F}_{q^{u_i}}$ divides $p_i$. 
Since $p_i$ is prime, the degree of $\alpha_{i,j_i}^e$ must be either $1$ or $p_i$.
If the degree of $\alpha_{i,j_i}^e$ over $\mathbb{F}_{q^{u_i}}$ is $1$, then $\alpha_{i,j_i}^e\in \mathbb{F}_{q^{u_i}}$, implying that $\alpha_{i,j_i}^{e(q^{u_i}-1)}=1$. 
Since $\alpha_{i,j_i}$ is a primitive element of $\mathbb{F}_{q^{p_i}}$, it follows that $q^{p_i}-1$ divides $e(q^{u_i}-1)$. Given that the greatest common divisor of $q^{p_i}-1$ and $q^{u_i}-1$ is $q-1$, we conclude that $\frac{q^{p_i}-1}{q-1}| e$.
This contradicts the assumption that $\frac{q^{p_i}-1}{q-1} \nmid e$. Consequently, the degree of $\alpha_{i,j_i}^e$ over $\mathbb{F}_{q^{u_i}}$ is $p_i$, implying that $\alpha_{i,j_i}^e$ serves as a defining element of $\mathbb{F}_{q^{u}}$ over $\mathbb{F}_{q^{u_i}}$. According to Lemma \ref{1}, there exists a subspace $S_{\mathcal{R}_{i}}$ of $\mathbb{F}_{q^{us}}$ over $\mathbb{F}_{q^{u_i}}$ such that
\begin{equation}\label{sr}
	\dim_{\mathbb{F}_{q^{u_i}}}S_{\mathcal{R}_{i}}=p_i,\quad S_{\mathcal{R}_{i}}+\alpha_{i,j_i}^eS_{\mathcal{R}_{i}}+\cdots+\alpha_{i,j_i}^{e(s-1)}S_{\mathcal{R}_{i}}=\mathbb{F}_{q^{us}}.
\end{equation}
Observing that $\mathbb{F}_{q^{u_i}}$ can be regarded as an extension field of $\mathbb{F}_{q^{\bar{u}_i}}$, by adjoining the elements $\alpha_{\kappa,j_{\kappa}}$ for all $\kappa \in [l] \setminus (\{i\} \cup \mathcal{R}_i)$ and $j_{\kappa} \in [t_{i_{\kappa}}]$. The extension degree of $\mathbb{F}_{q^{u_i}}$ over $\mathbb{F}_{q^{\bar{u}_i}}$ is $\frac{u_i}{\bar{u}_i}$. Consequently, a subspace of $\mathbb{F}_{q^{us}}$ over $\mathbb{F}_{q^{u_i}}$ with dimension $p_i$ is also a subspace of $\mathbb{F}_{q^{us}}$ over $\mathbb{F}_{q^{\bar{u}_i}}$ with dimension $\frac{u}{\bar{u}_i}$. Therefore, the conclusion follows from \eqref{sr}.
\end{IEEEproof}


\begin{remark}
	For completeness and theoretical insight, the theorem above presents a result stronger than the $e=1$ case required for single-node repair. 
	Specifically, it establishes a method to reconstruct any symbol in $\mathbb{F}_{q^{us}}$ via an $\mathbb{F}_{q^{\bar{u}_i}}$-linear combination of certain helper symbols.
	The rationale for operating over the field $\mathbb{F}_{q^{\bar{u}_i}}$ rather than $\mathbb{F}_{q^{u_i}}$ is to minimize the size of each required symbol. While the total number of bits transmitted during repair process is identical in both cases (\textit{i.e.}, over $\mathbb{F}_{q^{\bar{u}_i}}$ and $\mathbb{F}_{q^{u_i}}$), the probability of successful symbol transmission over $\mathbb{F}_{q^{\bar{u}_i}}$ exceeds that over $\mathbb{F}_{q^{u_i}}$ when the bit error probability is given. 
\end{remark}

For $i^{*} \in [l]$ and $j^{*}\in [t_{i^{*}}]$, if the node $\alpha_{i^{*},j^{*}}$ fails, we propose a repair scheme to recover the symbol stored in the failed node by communicating with helper nodes selected from the nodes excluded from the exclusion set $A_{i^{*}}$. This requires that the first subscript of each helper node differs from $i^{*}$.
We use the subscript index $ (i,j_i)$ of the evaluation point to index the corresponding node $\alpha_{i,j_i}$ of the storage system encoded by $RS(n,k,A)_E$.

\begin{theorem}\label{repair_scheme}
Let $\mathcal{C}$ be an $RS(n,k,A)_{E}$ code constructed from Construction~1.
Let $d$ be a positive integer satisfying $k\leq d\leq n-t$. The code $\mathcal{C}$ achieves the cut-set bound \eqref{cut-set-bound} for the repair of any single failed node, using any set of $d$ helper nodes selected from the nodes excluded from the exclusion set of the failed node.
\end{theorem}
\begin{IEEEproof}
Assume that the $(i^{*},j^{*})$-th node has been erased.
For a given value $ d $, let $ \mathcal{R}_{i^{*}} $ be a nonempty subset of $ [l] \setminus \{i^{*}\} $ such that $ |\mathcal{R}_{i^{*}}| \leq d \leq \sum_{i \in \mathcal{R}_{i^{*}}} t_i $, where $ |\mathcal{R}_{i^{*}}| $ denotes the cardinality of $ \mathcal{R}_{i^{*}} $.
The nodes indexed by $\mathcal{R}_{i^{*}}$ are denoted as $(i,j_i)$ with $i \in \mathcal{R}_{i^{*}}$ and $1 \leq j_i \leq t_i$, resulting in a total of $\sum_{i \in \mathcal{R}_{i^{*}}} t_i$ nodes.
A set of $d$ helper nodes is selected from the set $\left\{(i,j_i) : i \in \mathcal{R}_{i^{*}},\, 1 \leq j_i \leq t_i\right\}$ in such a way that the $z$-th helper node is chosen from the $1+(z-1)\bmod{|\mathcal{R}_{i^{*}}|}$-th group and is different from the previously chosen nodes.
This selection strategy ensures that the set of the first subscript index $i$ of the helper nodes is exactly the set $\mathcal{R}_{i^{*}}$. To clarify the presentation, without loss of generality, we may assume that the number $d$ of helper nodes satisfies $d = \sum_{i \in \mathcal{R}_{i^{*}}} t_i$.
According to Theorem~\ref{vector_space}, there exists a subspace $S_{\mathcal{R}_{i^{*}}}$ of $E$ over $\mathbb{F}_{q^{\bar{u}_{i^{*}}}}$ such that
\begin{align}
&\dim_{\mathbb{F}_{q^{\bar{u}_{i^{*}}}}}S_{\mathcal{R}_{i^{*}}}=\frac{u}{\bar{u}_{i^{*}}}, \label{s1}\\
&S_{\mathcal{R}_{i^{*}}}+\alpha_{i^{*},j^{*}}S_{\mathcal{R}_{i^{*}}}+\cdots+\alpha_{i^{*},j^{*}}^{s-1}S_{\mathcal{R}_{i^{*}}}=E, \label{s2}
\end{align}
where $\alpha S_{\mathcal{R}_{i^{*}}}=\{\alpha\beta:\beta\in S_{\mathcal{R}_{i^{*}}}\}$, and the operation $+$ is 
defined as $T_1+T_2=\{\beta_1+\beta_2:\beta_1\in T_1, \beta_2\in T_2\}$.
Let $h(x)$ be the annihilator polynomial of evaluation points $\alpha_{i,j_i}$ for all $i\notin\mathcal{R}_{i^{*}}$ and $1\leq j_i\leq t_i$ excluding the point $\alpha_{i^{*},j^{*}}$, that is
\[h(x)=\frac{\prod\limits_{i\notin\mathcal{R}_{i^{*}}}\prod\limits_{j_i\in[t_i]}(x-\alpha_{i,j_i})}{(x-\alpha_{i^{*},j^{*}})}.\]
Clearly, the degree of $h(x)$ is $n-d-1$. Let $g_w(x)=x^w h(x)$, where $w=0,1,\ldots,s-1$. It follows that the degree of $g_w(x)$ is less than or equal to $n-k-1$ for all $w=0,1,\ldots,s-1$ since
$s=d-k+1$. So for any $w\in \{0,1,\ldots,s-1\}$, $g_w(x)$ is a parity-check polynomial of the code $RS(n,k,A)_E$. Let
$(c_{1,1},\ldots,c_{1,t_1},c_{2,1},\ldots,c_{l,t_l})$ be a codeword of $RS(n,k,A)_E$. Recall that the dual code of
$RS(n,k,A)_E$ is the GRS code $GRS(n,n-k,A,\mathbf{v})_E$, where $\mathbf{v}\in (E^{*})^{n}$ are all nonzero coordinates. We ignore the nonzero coefficients in the repair processing which do not affect the repair process. Therefore, we have
\begin{equation}\label{check1}
\sum_{i\in[l]}\sum_{j_i\in t_i}\alpha_{i,j_i}^{w}h(\alpha_{i,j_i})c_{i,j_i}=0 \text{ for all } w=0,1,\ldots,s-1.
\end{equation}
Let $e_1,e_2,\ldots,e_{\frac{u}{\bar{u}_{i^{*}}}}$ be an arbitrary basis of the subspace $S_{\mathcal{R}_{i^{*}}}$ over the field $\mathbb{F}_{q^{\bar{u}_{i^{*}}}}$. From (\ref{check1}) we obtain the following $\frac{su}{\bar{u}_{i^{*}}}$ equations:
\begin{equation}\label{check2}
\sum_{i\in[l]}\sum_{j_i\in t_i} e_m\alpha_{i,j_i}^{w}h(\alpha_{i,j_i})c_{i,j_i}=0,
\end{equation}
for all $w=0,1,\ldots,s-1$ and $ m=1,2,\ldots,\frac{u}{\bar{u}_{i^{*}}}$.
Let $Tr_{i^{*}}=Tr_{E/\mathbb{F}_{q^{\bar{u}_{i^{*}}}}}$ be the trace map onto the subfield $\mathbb{F}_{q^{\bar{u}_{i^{*}}}}$. Performing $Tr_{i^{*}}$ to equation (\ref{check2}) yields
\begin{equation}\label{check3}
\sum_{i\in[l]}\sum_{j_i\in t_i}Tr_{i^{*}}(e_m\alpha_{i,j_i}^{w}h(\alpha_{i,j_i})c_{i,j_i})=0
\end{equation}
for all $w=0,1,\ldots,s-1$ and $m=1,2,\ldots,\frac{u}{\bar{u}_{i^{*}}}$.
Applying the annihilator property of $h(x)$ to \eqref{check3}, we obtain
\begin{align}
Tr_{i^{*}}(e_m\alpha_{i^{*},j^{*}}^{w}h(\alpha_{i^{*},j^{*}})c_{i^{*},j^{*}}) 
&=-\sum_{i\in\mathcal{R}_{i^{*}}}\sum_{j_i\in t_i}Tr_{i^{*}}(e_m\alpha_{i,j_i}^{w}h(\alpha_{i,j_i})c_{i,j_i}) \\
& =-\sum_{i\in\mathcal{R}_{i^{*}}}\sum_{j_i\in t_i} \alpha_{i,j_i}^{w} Tr_{i^{*}}(e_m h(\alpha_{i,j_i})c_{i,j_i})\label{12}
\end{align}
for all $w=0,1,\ldots,s-1$ and $m=1,2,\ldots,\frac{u}{\bar{u}_{i^{*}}}$, where the equality \eqref{12} follows from the fact that the trace map $Tr_{i^{*}}$ is $\mathbb{F}_{q^{\bar{u}_{i^{*}}}}$-linear and $\alpha_{i,j_i}^{w}\in \mathbb{F}_{q^{\bar{u}_{i^{*}}}}$ for all $i\in\mathcal{R}_{i^{*}}$, $j_i\in t_i$ and $w=0,1,\ldots,s-1$.
According to Theorem~\ref{vector_space}, the set $$\{e_m\alpha_{i^{*},j^{*}}^{w}: 0\leq w\leq s-1,\, 1\leq m\leq \frac{u}{\bar{u}_{i^{*}}}\}$$ forms a basis of $E$ over $\mathbb{F}_{q^{\bar{u}_{i^{*}}}}$ and so does the set $$\{e_m\alpha_{i^{*},j^{*}}^{w}h(\alpha_{i^{*},j^{*}}): 0\leq w\leq s-1,\, 1\leq m\leq \frac{u}{\bar{u}_{i^{*}}}\},$$ since $h(\alpha_{i^{*},j^{*}})\neq 0$. 
Therefore, similar to the Guruswami-Wootters' repair scheme \cite{guruswami2017}, 
the symbol $c_{i^{*},j^{*}}$ can be reconstructed from the parity-check equations presented in~\eqref{12} by utilizing the property of the trace dual basis presented in \eqref{dual_bases}.
To recover $c_{i^{*},j^{*}}$, the repair process downloads the $\frac{u}{\bar{u}_{i^{*}}}$ symbols of $\mathbb{F}_{q^{\bar{u}_{i^{*}}}}$ from each helper node indexed by $\mathcal{R}_{i^{*}}$, which is equivalent to downloading a total of $ud$ $\mathbb{F}_q$-symbols.
Therefore, the repair bandwidth of the repair scheme achieves the cut-set bound~\eqref{cut-set-bound}.
\end{IEEEproof}

In the following, we present an explicit $(12,8)$ RS code over $\mathbb{F}_{2^{2310}}$ and exhibit the detailed repair scheme for the single-node repair.
Notably, the recovery computation is executed on a standard computer (Intel Core i9-14900HX, 64 GB RAM, Windows 11) with negligible latency.
\begin{example}\label{example1}
	Let $s=2$, $p_1=3$, $p_2=5$, $p_3=7$ and $p_4=11$. 
	Let $E=\mathbb{F}_{2^{2310}}$ be an extension field of $\mathbb{F}_2$ with a generator $g$, where the minimal polynomial of $g$ over $\mathbb{F}_2$ is $x^{2310}+x^8+x^5+x^2+1$. Let $F=\mathbb{F}_{2^{1155}}$ and $K=\mathbb{F}_{2^{385}}$ be subfields of $E$. It can be verified that $g$ is a defining element of $E$ over $K$. Let $F_1=\mathbb{F}_{2^3}$, $F_2=\mathbb{F}_{2^5}$, $F_3=\mathbb{F}_{2^7}$ and $F_4=\mathbb{F}_{2^{11}}$ be subfields of $E$ with primitive elements $g_1$, $g_2$, $g_3$ and $g_4$, respectively. The minimal polynomials of $g_1$, $g_2$, $g_3$ and $g_4$ are $x^3+x^2+1$, $x^5+x^4+x^3+x+1$, $x^7+x^6+x^5+x^2+1$ and $x^{11}+x^9+x^7+x^4+x^3+x^2+1$, respectively. 
	Define four sets:
	\begin{align*}
		A_1&=\{g_1, g_1^2, g_1^3\}, A_2=\{g_2, g_2^2, g_2^3\},\\
	    A_3&=\{g_3, g_3^2, g_3^3\}, A_4=\{g_4, g_4^2, g_4^3\}.
	\end{align*}
	For each primitive element $g_{i}$ of $F_i$, where $1\leq i\leq 4$, the power $g_{i}^{k}$ is a primitive element of $F_i$ if and only if $\gcd(k, |F_i|-1)=1$. Verify that each element in $A_i$ is a primitive element of $F_i$ for all $i = 1, \ldots, 4$. Let $A$ be the union set of $A_1,A_2,A_3$ and $A_4$.
	 Define an RS code over $E$ as follows:
	\[RS(12,8,A)=\left\{(f(\alpha), \alpha\in A)\, : \, f(x)\in E[x],\, \deg(f)<8 \right\}.\]
	We present an explicit repair scheme of the code $RS(12,8,A)$ that achieves the cut-set bound for a single failed node. Without loss of generality, assume that the first node $g_1$ of $A_1$ has failed. The repair of node $g_1$ requires downloading data from the nodes in a set $\mathcal{R}_{g_1}$, which is set as the union of sets $A_2$, $A_3$, and $A_4$. Note that the dual code of $RS(12,8,A)$ is the code 
	$GRS(12,4,A,v)$, where $v\in (E^{*})^{12}$ can be determined by
	\begin{equation}\label{v}
		v_{i}=\prod_{\alpha_j\in A\setminus \{\alpha_i\}}(\alpha_i-\alpha_j)^{-1},
	\end{equation}
	where $\alpha_i$ and $\alpha_j$ denote different evaluation points of $A$ for $1\leq j\neq i\leq 12$. The following outlines the step-by-step procedure for repairing the failed node $g_1$.
	\begin{enumerate}
		\item[Step 1.] Determine a subspace $S_1$ of $E$ and a basis for $S_1$ over $K$, where $S_1$ corresponds to $S_{\mathcal{R}_{i^{*}}}$ in \eqref{s1}. Let $$S_1=\text{Span}_{K}(gg_1, g^2g_{1}^{2}, g_{1}^{2}+gg_{1}^{2}).$$ 
		Verify that $S_1$ is a three-dimensional subspace of $E$ over $K$, and that $gg_1, g^2g_{1}^{2}, g_{1}^{2}+gg_{1}^{2}$ are $K$-linearly independent. Hence, the set $\{gg_1, g^2g_{1}^{2}, g_{1}^{2}+gg_{1}^{2}\}$ forms a basis of $S_1$.
		Verify that $\text{Span}_{K}(gg_{1}, g^{2}g_{1}^{2}, g_{1}^{2}(1+g), gg_{1}^{2}, g^{2}g_{1}^{3}, g_{1}^{3}(1+g))$ forms a full vector space $E$, a result consistent with Theorem~\ref{vector_space}.
		Therefore, the set 
		\begin{equation}\label{basis}
			\{gg_{1}, g^{2}g_{1}^{2}, g_{1}^{2}(1+g), gg_{1}^{2}, g^{2}g_{1}^{3}, g_{1}^{3}(1+g)\}
		\end{equation}
	  forms a basis of $E$ over $K$.
	  \item[Step 2.] Download symbols from helper nodes $\mathcal{R}_{g_1}$. Let $h(x)=(x-g_{1}^{2})(x-g_{1}^{3})$. Since $h(g_1)$ and $v_1$ are nonzero elements, it follows from \eqref{basis} that $$B=\{gg_{1}h(g_{1})v_1, g^{2}g_{1}^{2}h(g_{1})v_1, g_{1}^{2}(1+g)h(g_{1})v_1, gg_{1}^{2}h(g_{1})v_1, g^{2}g_{1}^{3}h(g_{1})v_1, g_{1}^{3}(1+g)h(g_{1})v_1\}$$ also forms a basis of $E$ over $K$. For simplicity, we denote the elements of set $B$ as $b_1, b_2, b_3, b_4, b_5,$ and $b_6$. The replacement node of the failed node $g_1$ needs to download three $K$-symbols from each helper node $g_{i}^j$, where $i \in \{2, 3, 4\}$ and $j \in \{1, 2, 3\}$. The symbols downloaded from each helper node $g_i^j$ are listed in the following table.
		\begin{table}[h]
			\centering
			\caption{Download Three $K$-symbols from Each Helper Node $g_{i}^j$, where $i \in \{2, 3, 4\}$ and $j \in \{1, 2, 3\}$}
			\begin{tabular}{|l|l|}
				\hline
				The First $K$-symbol $\beta_{i,j}^{(1)} $  & $Tr_{E/K}(gg_1 h(g_{i}^j) f(g_{i}^j))$ \\ 
				The Second $K$-symbol $\beta_{i,j}^{(2)}$   & $Tr_{E/K}(g^2 g_1^2 h(g_{i}^j) f(g_{i}^j))$ \\ 
				The Third $K$-symbol  $\beta_{i,j}^{(3)}$  & $Tr_{E/K}((g_1^2 + g g_1^2) h(g_{i}^j) f(g_{i}^j))$ \\ \hline
			\end{tabular}
		\end{table}
		
		\item[Step 3.] Compute the trace dual basis of $B$ over $K$. Compute a trace matrix $A$, defined as  
		\[ T=\left[ \begin{matrix}
			T{{r}_{E/K}}({{b}_{1}}{{b}_{1}}) & T{{r}_{E/K}}({{b}_{1}}{{b}_{2}}) & \cdots  & T{{r}_{E/K}}({{b}_{1}}{{b}_{6}})  \\
			T{{r}_{E/K}}({{b}_{2}}{{b}_{1}}) & T{{r}_{E/K}}({{b}_{2}}{{b}_{2}}) & \cdots  & T{{r}_{E/K}}({{b}_{2}}{{b}_{6}})  \\
			\vdots  & \vdots  & \ddots  & \vdots   \\
			T{{r}_{E/K}}({{b}_{6}}{{b}_{1}}) & T{{r}_{E/K}}({{b}_{6}}{{b}_{2}}) & \cdots  & T{{r}_{E/K}}({{b}_{6}}{{b}_{6}})  \\
		\end{matrix} \right].
		\]	
		Denote $I_6$ be the identity matrix of order $6$ over $K$. Consider the system of linear equations 
		\begin{equation}\label{equ}
			TX=I_6,
		\end{equation}
		solve the unique solution $X=[x_{ij}]_{1\leq i,j\leq 6}$. Denote the dual basis of $B$ as $B^{\perp}=\{b_i^{\perp}| 1\leq i\leq 6\}$.
		Then 
		\begin{equation}\label{perp}
			b_i^{\perp}=\sum_{j=1}^{6}x_{ij}b_j.
		\end{equation}

		\item[Step 4.] Recover the symbol $f(g_{1})$.
		As shown in \eqref{12},  by computing six summations over the downloaded symbols $\beta_{i,j}^{(m)}$:  
		\[Tr_{E/K}(b_{m(w+1)}f(g_{1})))=-\sum_{i=2}^4\sum_{j=1}^3 g_{i}^{jw} \beta_{i,j}^{(m)},\]
		where $w=0,1$ and $m=1,2,3$,
		the replacement node obtains six $K$-symbols:
		\begin{equation}\label{recover}
			Tr_{E/K}(b_{1}f(g_{1}))), \cdots, Tr_{E/K}(b_{6}f(g_{1}))),
		\end{equation}
		where $\{b_{1}, b_{2},\ldots, b_{6}\} \in B$. 
		Therefore, by combining \eqref{perp} and \eqref{recover}, we obtain $$\tilde{f} = \sum_{i=1}^{6} Tr_{E/K}(b_i f(g_1)) b_i^{\perp}.$$  
		The replacement node successfully recovers the erased symbol $f(g_1)$ for  $f(g_1)= \tilde{f}$.
	\end{enumerate}
	
	The repair bandwidth of recovering one $E$-symbol is $27$ $K$-symbols, achieving the cut-set bound~\eqref{cut-set-bound}. More precisely, reconstructing $2310$ bits of data requires the transmission of 10395 bits in total, which is significantly lower than the repair bandwidth of the naive repair approach, which amounts to 18480 bits. 
	Suppose $f(x)=x^3+x^2+1$. We executed the Magma program \cite{magma97} to confirm that the result $\tilde{f}$ computed in Step 4 is consistent with the original symbol $f(g_1)$.
\end{example}

\begin{theorem}\label{sub_packetization}
	The sub-packetization level of the constructed codes in Construction~1 is $[E:\fq]=s\prod_{i=1}^{l}p_i$, where the $p_i$'s are the smallest distinct primes satisfying~\eqref{p}, $l$ is the number of exclusion sets of the associated PE repair scheme.
\end{theorem}

\begin{remark}\label{3}
	We compare our constructions with these existing ones that achieve or asymptotically approach the cut-set bound. 
	We provide the explicit typical $(12,8)$ RS codes derived from Construction 1 and those from the related constructions in \cite{Chowdhury2021ImprovedSF,ye2016explicit,tamo2017optimal,tamo2019repair} in Table~\ref{comparison}. This type of RS codes is employed in industry cloud storage \cite{lai2015atlas}. 
	As shown in Table~\ref{comparison}, our constructions achieve significantly lower sub-packetization levels, which enables the repair scheme to operate on regular commercial computers.
	
\end{remark}

For an $(n,k)$ scalar MDS code with a conventional linear repair scheme, if the repair bandwidth of the scheme achieves the cut-set bound for the repair of any single
node from any $d$ helper nodes for $k\leq d\leq n-1$, it has been proven in \cite[Theorem~8]{tamo2019repair} that
the sub-packetization level $L$ is at least
\begin{equation}\label{L}
	L\geq \prod_{i=1}^{k-1} p_i.
\end{equation}

\begin{remark}\label{2}
	According to Theorem~\ref{sub_packetization}, if the dimension $k$ is strictly greater than $l+3$, then the sub-packetization level in our construction is strictly smaller than the the lower bound presented in~\eqref{L} for conventional repair schemes. This indicates that, from the perspective of sub-packetization level for RS codes achieving the cut-set bound, our construction outperforms all the existing and potential constructions designed for conventional repair schemes. 
\end{remark}

\begin{example}
	For the $(12,8)$ RS code presented in Example \ref{example1}, the achieved sub-packetization level is $2310$.  In contrast, the lower bound of the sub-packetization level for any $(12,8)$ RS code under the conventional repair framework is $L\geq 510510$. 
	The exceedingly high lower bound demonstrates that implementing such codes under the conventional repair framework is impractical.
	In comparison, the achieved sub-packetization level of Construction 1 can be significantly lower than that of constructions for conventional repair schemes. 
	The achievable sub-packetization level enable the efficient implementation of the PE repair scheme for $(12,8)$ RS codes on a regular commercial computer.
\end{example}
\section{The Second Construction of RS codes with Optimal Repair Bandwidth}\label{construction2}
Construction 1 enables the same value of $t_i$ for every $i$, a choice that can raise the code rate given the relation $k \leq n - \max_{i\in[l]} t_i$ with $n=\sum_{i=1}^{l}t_i$. At the same time, the construction imposes the condition $p_i \equiv 1 \pmod{s}$ for all $i$, which in turn yields a high sub-packetization level, particularly when $s > 2$.
To eliminate this requirement and further reduce the sub-packetization level, we propose a new construction, albeit with stricter constraints on the values of $t_i$ for each $i\in [l]$.

\noindent \textbf{Construction 2.} Let $\mathbb{F}_q$ be a finite field. For a positive integer $l\geq 2$, let $r$ be a positive integer and let $p_1,p_2,\ldots,p_l$ be $l$ prime numbers such that for each $1\leq i\leq l$,
\begin{equation}\label{con_p}
	\phi(q^{p_i}-1)\geq r-p_i+1\geq 2,
\end{equation}
where $\phi$ is the Euler's totient function.
Let 
\begin{equation}\label{t}
	t_i=r-p_i+1,
\end{equation}
where $1\leq i\leq l$. For each prime extension field $\mathbb{F}_{q^{p_i}}$ of $\mathbb{F}_q$, let $\alpha_{i,j_i}$ be a primitive element of $\mathbb{F}_{q^{p_i}}$, where $1\leq j_i\leq t_i$. Note that for each extension field $\mathbb{F}_{q^{p_i}}$, it is possible to select $t_i$ distinct primitive elements of $\mathbb{F}_{q^{p_i}}$, since $t_i\leq \phi(q^{p_i}-1)$.
For each $1\leq i\leq l$, let the $i$-th exclusion set be
\begin{equation}
	A_i=\{\alpha_{i,j_i}\,:\,1\leq j_i\leq t_i\}.
\end{equation}
Define two notations as follows:
\begin{equation}\label{n-u}
	u_i = \prod_{j \in [l] \setminus \{i\}} p_j, \quad u = \prod_{j \in [l]} p_j.
\end{equation}
Let the evaluation point set be $A=\bigcup_{i=1}^{l}A_i$. Let $n=\sum_{i=1}^{l}t_i$ and $k=n-r$. Let $E$ be the algebraic extension field of $\mathbb{F}_q$ by adjoining the elements of $A$. 
Define the  code $RS(n,k,A)_{E}$ as
\[\left\{(f(\alpha_{i,j_i}),1 \leq i \leq l, 1 \leq j_i \leq t_i)\, : \, f(x)\in E[x],\, \deg(f)<k \right\}.\]

We use the evaluation point to index the node. We propose an optimal repair scheme of $RS(n,k,A)_{E}$ for repairing any single failed node by leveraging partially information of the nodes in groups different from the group that contains the failed node.
Notably, compared to Construction 1, Construction 2 reduces the sub-packetization level of RS codes, albeit at the cost of reduced flexibility of the parameters $t_i$ for all $i\in[l]$.
\begin{lemma}\label{lemma7}
	For each evaluation point $\alpha_{i,j_i}\in A$, where $i$ and $j_i$ are fixed integers such that $i\in[l] $ and $j_i\in [t_i]$, let $F_i$ be the algebraic extension field of $\mathbb{F}_q$ that containing $\alpha_{i^{'},j_{i^{'}}}\in A$ for all $i^{'}\in [l]\setminus\{i\}$ and $j_{i^{'}}\in [t_{i^{'}}]$. Then $\{1,\alpha_{i,j_i},\ldots,\alpha_{i,j_i}^{p_i-1}\}$ is a basis of $E$ over $F_i$.
\end{lemma}
\begin{IEEEproof}
	It is a straightforward consequence of the properties of finite fields. We include the proof here for completeness. We first demonstrate that $F_i$ is a subfield of $E$ with cardinality $q^{u_i}$, where $u_i$ is defined as in \eqref{n-u}. Assume that $F_i$ has cardinality $q^{n}$. On one hand, since $p_{i^{'}}$ divides $u_i$ for each $i^{'}\in [l]\setminus\{i\}$, it follows from the subfield criterion theorem that $\alpha_{i^{'},j_{i^{'}}}$ is an element in $\mathbb{F}_{q^{u_i}}$ for all $i^{'}\in [l]\setminus\{i\}$ and $j_{i^{'}}\in [t_{i^{'}}]$. It follows that $\mathbb{F}_{q^{u_i}}$ is an extension field of $F_i$. By the subfield criteria theorem, we have $n$ divides $u_i$. 
	On the other hand, since $\fq(\alpha_{i,j_i})=\mathbb{F}_{q^{p_i}}$ is a subfield of $F_i$, it follows that $q^{p_{i^{'}}}$ divides $q^{n}$ for all $i^{'}\in [l]\setminus\{i\}$, which implies that $p_{i^{'}}$ divides $n$ for all such $i^{'}$. Given that $p_i$ are distinct prime numbers for all $i\in[l]$, and $u_i = \prod_{j \in [l] \setminus \{i\}} p_j$, it follows that $u_i$ divides~$n$. Consequently, we obtain that $n$ equals to $u_i$ and $F_i=\mathbb{F}_{q^{u_i}}$.
	
   For each $i\in [l]$ and $j_{i}\in [t_{i}]$, consider the simple algebraic extension field $\mathbb{F}_{q^{u_i}}(\alpha_{i,j_i})$ of $\mathbb{F}_{q^{u_i}}$. The cardinality of $\mathbb{F}_{q^{u_i}}(\alpha_{i,j_i})$ is $q^{u_i[\mathbb{F}_{q^{u_i}}(\alpha_{i,j_i}):\mathbb{F}_{q^{u_i}}]}$, which divides $q^u$ since $\mathbb{F}_{q^{u_i}}(\alpha_{i,j_i})$ is a subfield of $E$. Therefore, $[\mathbb{F}_{q^{u_i}}(\alpha_{i,j_i}):\mathbb{F}_{q^{u_i}}]$ divides $p_i$. Since $p_i$ is a prime, it follows that $[\mathbb{F}_{q^{u_i}}(\alpha_{i,j_i}):\mathbb{F}_{q^{u_i}}]$ is $1$ or $p_i$. Since $\alpha_{i,j_i}\notin \mathbb{F}_{q^{u_i}}$, it follows that $[\mathbb{F}_{q^{u_i}}(\alpha_{i,j_i}):\mathbb{F}_{q^{u_i}}]$ is $p_i$. According to Lemma \ref{5} (ii), we obtain $\{1,\alpha_{i,j_i},\ldots,\alpha_{i,j_i}^{p_i-1}\}$ is a basis of $E$ over $F_i$.
\end{IEEEproof}

\begin{theorem}\label{repair_scheme2}
	The code $RS(n,k,A)_{E}$, presented in the Construction 2, achieves the cut-set bound \eqref{cut-set-bound} for the repair of any single failed node, by downloading symbols from all surviving nodes that are excluded from the exclusion set of the failed node.
\end{theorem}
\begin{IEEEproof}
	Assume that the $(i^{*},j^{*})$-th node has been erased.
	Denote the set of the helper nodes responsible for repairing the $(i^{*},j^{*})$-th node as 
	$$\mathcal{R}=\left\{(i,j_i) : i \in [l] \setminus \{i^{*}\},\, 1 \leq j_i \leq t_i\right\}.$$ 
	Let $h(x)$ be the annihilator polynomial of evaluation points $A_{i^{*}}$ apart from the point $\alpha_{i^{*},j^{*}}$, that is
	\[h(x)=\prod\limits_{j_{i^{*}}\in[t_{i^{*}}]\setminus \{j^{*}\}}(x-\alpha_{i^{*},j_{i^{*}}}).\]
	Clearly, the degree of $h(x)$ is $t_{i^{*}}-1$. Let $g_w(x)=x^w h(x)$, where $w=0,1,\ldots,p_{i^{*}}-1$. It follows that the degree of $g_w(x)$ is less than or equal to $p_{i^{*}}+t_{i^{*}}-2$, which equals to $r-1$ according to the definition of $t_i$ in \eqref{t}. So for any $w\in \{0,1,\ldots,p_{i^{*}}-1\}$, $g_w(x)$ is a parity-check polynomial of the code $RS(n,k,A)_E$. Let
	$(c_{1,1},\ldots,c_{1,t_1},c_{2,1},\ldots,c_{l,t_l})$ be a codeword of $RS(n,k,A)_E$. Recall that the dual code of
	$RS(n,k,A)_E$ is the GRS code $GRS(n,n-k,A,\mathbf{v})_E$, where the coordinates of $\mathbf{v}\in (E^{*})^{n}$ are all nonzero. We ignore the nonzero coefficients in the repair processing which do not affect the repair process. Therefore, we have
	\begin{equation}\label{check5}
		\sum_{i\in[l]}\sum_{j_i\in t_i}\alpha_{i,j_i}^{w}h(\alpha_{i,j_i})c_{i,j_i}=0 \text{ for all } w=0,1,\ldots,p_{i^{*}}-1.
	\end{equation}
	Let $Tr_{i^{*}}=Tr_{E/\mathbb{F}_{q^{u_{i^{*}}}}}$ be the trace map onto the subfield $\mathbb{F}_{q^{u_{i^{*}}}}$. Performing the trace map to (\ref{check5}) yields
	\begin{equation}\label{check6}
		\sum_{i\in[l]}\sum_{j_i\in t_i}Tr_{i^{*}}(\alpha_{i,j_i}^{w}h(\alpha_{i,j_i})c_{i,j_i})=0
	\end{equation}
	for all $w=0,1,\ldots,p_{i^{*}}-1$.
	Applying the annihilator property of $h(x)$ to equation \eqref{check6}, we obtain
	\begin{align}
		Tr_{i^{*}}(\alpha_{i^{*},j^{*}}^{w}h(\alpha_{i^{*},j^{*}})c_{i^{*},j^{*}}) 
		&=-\sum_{i\in[l] \setminus \{i^{*}\}}\sum_{j_i\in t_i}Tr_{i^{*}}(\alpha_{i,j_i}^{w}h(\alpha_{i,j_i})c_{i,j_i}) \\
		& =-\sum_{i\in[l] \setminus \{i^{*}\}}\sum_{j_i\in t_i} \alpha_{i,j_i}^{w} Tr_{i^{*}}( h(\alpha_{i,j_i})c_{i,j_i})\label{check7}
	\end{align}
	for all $w=0,1,\ldots,p_{i^{*}}-1$, where the equality \eqref{check7} follows from the fact that the trace map $Tr_{i^{*}}$ is $\mathbb{F}_{q^{u_{i^{*}}}}$-linear and $\alpha_{i,j_i}^{w}\in \mathbb{F}_{q^{u_{i^{*}}}}$ for all $i\in[l] \setminus \{i^{*}\}$, $j_i\in t_i$ and $w=0,1,\ldots,p_{i^{*}}-1$.
	Since $\alpha_{i^{*},j^{*}}$ is a primitive element of $\mathbb{F}_{q^{p_i}}$, it follows from Lemma~\ref{lemma7} that 	
	the set $$\{\alpha_{i^{*},j^{*}}^{w}: 0\leq w\leq p_{i^{*}}-1\}$$ forms a basis of $E$ over $\mathbb{F}_{q^{u_{i^{*}}}}$ and so does the set $$\{\alpha_{i^{*},j^{*}}^{w}h(\alpha_{i^{*},j^{*}}): 0\leq w\leq p_{i^{*}}-1\},$$ since $h(\alpha_{i^{*},j^{*}})\neq 0$. Similar to the proof of Theorem \ref{repair_scheme}, the symbol $c_{i^{*},j^{*}}$ can be reconstructed from the parity-check equations presented in~\eqref{check7} by utilizing the property of the trace dual basis presented in \eqref{dual_bases}.

	To recover $c_{i^{*},j^{*}}$, the repair process downloads one single symbol of $\mathbb{F}_{q^{u_{i^{*}}}}$ from each helper node indexed by $\mathcal{R}_{i^{*}}$. Therefore, the repair bandwidth is $(n-t_{i^{*}})u_{i^{*}}$ symbols of $\mathbb{F}_q$, which meets the cut-set bound $\frac{(n-t_{i^{*}})u}{n-t_{i^{*}}-k-1}$, since $n-t_{i^{*}}-k-1=p_{i^{*}}$ and $u= u_{i^{*}}p_{i^{*}}$.
\end{IEEEproof}

In the following, we present an explicit $(17,9)$ RS code over $\mathbb{F}_{2^{60}}$ and exhibit the detailed repair scheme for the single-node repair.
\begin{example}\label{example2}
	Let $q=4$, $r=8$, $l=3$, $p_1=2$, $p_2=3$ and $p_3=5$. Let $g_1,g_2$ and $g_3$ be primitive elements of $\mathbb{F}_{4^2},\mathbb{F}_{4^3}$ and $\mathbb{F}_{4^5}$ whose primitive polynomials over $\mathbb{F}_2$ are $x^4+ x + 1, x^6 + x^4 + x^3+x+ 1$ and $x^{10} + x^6 + x^5 + x^3+ x^2 + x + 1$, respectively.
	Define three exclusion sets of evaluation points as follows:
	\begin{align*}
		A_1&=\{g_1,g_1^2,g_1^4,g_1^7,g_1^8,g_1^{11},g_1^{13}\},\\
		A_2&=\{g_2,g_2^2,g_2^4,g_2^5,g_2^8,g_2^{10}\},\\
		A_3&=\{g_3,g_3^2,g_3^4,g_3^5\}.
	\end{align*}
	Note that all the elements in these three sets are primitive elements over the corresponding finite fields. Denote $t_1=7$, $t_2=6$, $t_3=4$. For simplicity, we use the notation $g_{i,j_i}$ to denote the $j_i$-th evaluation point of the group $A_i$, where $1\leq i\leq 3$ and $1\leq j_i\leq t_i$.
	Let $A$ be the union set of $A_i$ for $1\leq i\leq 3$. Let $E=\mathbb{F}_{4^{30}}$ be the extension field of $\mathbb{F}_{4}$ containing $\mathbb{F}_{4^2},\mathbb{F}_{4^3}$ and $\mathbb{F}_{4^5}$ as subfields.
	Define a $(17,9)$ RS code over $E$ as follows:
	\begin{equation*}
		RS(17,9,A)=\left\{(f(g_{i,j_i}),g_{i,j_i} \in A \text{ for }1 \leq i \leq 3, 1 \leq j_i \leq t_i)\, : \, f(x)\in E[x],\, \deg(f)<9 \right\}.
	\end{equation*}
	Suppose that $f(g_{1,1})=f(g_1)$ has been erased. Let $K$ be a subfield of $E$ with cardinality $4^{15}$. We present an optimal repair scheme to exactly recover the symbol $ f(g_1)$ by downloading a single $K$- symbol from the helper nodes in $ A_2, A_3$.
	\begin{itemize}
		\item[ Step 1.]  Determine the dual code of the RS code. The dual code of $RS(17,9,A)$ is the code $GRS(17,8,A,v)$, where $v=(v_{1,1},\ldots,v_{1,7},\ldots,v_{3,4})\in (E^*)^{17}$ can be determined similar to \eqref{v}.
		\item[Step 2.] Download a single $K$-symbol from each helper node $g_{i,j_i}$ for $2\leq i\leq 3$ and $1\leq j_i\leq t_i$. Let $h(x)=(x-g_1^2)(x-g_1^4)(x-g_1^7)(x-g_1^8)(x-g_1^{11})(x-g_1^{13})$. The download symbol is $$Tr_{E/K}(h(g_{i,j_i})c(g_{i,j_i})v_{i,j_i}).$$
		\item[Step 3.]  Compute the trace dual basis of $\{h(g_{1})v_{1,1}, g_1h(g_{1})v_{1,1}\}$. According to Lemma \ref{lemma7}, the set $\{1,g_1\}$ is a basis of $E$ over $K$. Since 
		$h(g_{1})$ and $v_1$ are nonzero element in $E$, it follows that $\{h(g_{1})v_{1,1}, g_1h(g_{1})v_{1,1}\}$ is also a basis of $E$ over $K$. Similar to Step 3 of Example~\ref{example1}, the dual basis $\{h(g_{1})v_{1,1}, g_1h(g_{1})v_{1,1}\}$, which we denote as $B^{\perp}=\{b_1^{\perp},b_2^{\perp}\}$, can be determined after solving a system of linear equations similar to \eqref{equ}.
		\item[Step 4.] Recover $f(g_{1})$. According to \eqref{check7}, the replacement node perform summations of the downloaded symbols:
		$$Tr_{E/K}(g_1^{w} h(g_{1})v_{1,1}f(g_{1}))=-\sum_{i\in\{2,3\}}\sum_{j_i\in t_i} \alpha_{i,j_i}^{w} Tr_{E/K}( h(\alpha_{i,j_i})v_{i,j_i}f_{i,j_i}), \text{ for } w=0,1.$$
		Since $\{h(g_{1})v_{1,1}, g_1h(g_{1})v_{1,1}\}$ forms a basis of $E$ over $F$, the repaired symbol $\tilde{f}$ is equal to 
		\begin{equation}\label{c}
			Tr_{E/K}(h(g_{1})v_{1,1}f(g_{1}))b_1^{\perp}+Tr_{E/K}(g_1h(g_{1})v_{1,1}f(g_{1}))b_2^{\perp}.
		\end{equation}
	\end{itemize}

	The repair bandwidth required to recover one $E$-symbol varies across different nodes, with each node achieving the cut-set bound~\eqref{cut-set-bound}.
	Suppose that the encoder polynomial is $f(x)=x^3+x^2+x+1$. We implemented the repair procedure and verified via Magma \cite{magma97} that the repaired symbol given by \eqref{c} coincides exactly with the original symbol $f(g_{1})$.
\end{example}

\begin{remark}\label{comparison2}
Construction 2 lowers the sub-packetization level relative to Construction 1, while requires the stricter parameter condition $t_i=n-k-p_i+1$ for all $1\leq i\leq l$, which restricts the achievable RS code parameters. 
Notably, in Construction~2, 
the repair locality for each node varies across different groups. This leads to the nodes from different exclusion sets exhibit different repair bandwidths.
For the explicit $(17,9)$ RS codes presented in Example~\ref{example2}, nodes in $A_1$ have the repair locality $d_1=10$, those in $A_2$ have $d_2=11$, and those in $A_3$ have $d_3=13$, giving an average
$$d=\frac{|A_1|d_1+|A_2|d_2+|A_3|d_3}{|A_1|+|A_2|+|A_3|}>11.$$ For consistency, Table~\ref{comparison_table} lists explicit codes constructed via \cite{tamo2017optimal,tamo2019repair} and Construction~1, all of which have $d = 11$. To meet the requirement $d = 11$, the flexibility of the $(17,9)$ RS code constructed via Construction~1 is set as $t=n-d=6$.
As shown in Table~\ref{comparison_table}, even when Construction 2 maintains a lower average normalized repair bandwidth per node,
it achieves a significantly lower sub-packetization level than both Construction 1 and the existing counterpart constructions.

\end{remark}
	\begin{table}[ht]
	\centering
	\caption{Comparison of Explicit $(17,9)$ RS codes that Attain the Cut-Set Bound over Base Field $\mathbb{F}_4$}
	\resizebox{\textwidth}{!}{
		\begin{tabular}{ccccc}
			\toprule[0.8pt]
			Constructions  & Sub-Packetization Level & Repair Bandwidth (Bits)$^{\mathrm{a}}$  & Repair Locality & Normalized Repair Bandwidth$^{\mathrm{b}}$ \\
			\midrule
			\cite{tamo2017optimal,tamo2019repair}  &  $\approx 2.75\times 10^{30}$ & $\approx 2.02\times 10^{31} $ & 11 &   $11/3$ \\
			Construction 1 (Sec. IV)  &  $ 5187 $ & $ 38038$  &  11 & $11/3$ \\
			 & & Nodes in $A_1$: 300 & 10 & 5\\
			\multirow{-1}{*}{Construction 2 (Sec. V)}   & \multirow{-1}{*}{$30$} &Nodes in $A_2$: 220 & 11 & \multirow{-1}{*}{$11/3$}\\
			 & &  Nodes in $A_3$: 156 & 13 & 2.6\\
			\bottomrule[0.8pt]
		\end{tabular}
	}
	\label{comparison_table}
	
	\vspace{0.1cm}
	\footnotesize
	\raggedright
	\textbf{Note:} $\mathrm{a}$. The repair bandwidth column measures the total number of bits required to repair a failed node; $\mathrm{b}$. The normalized repair bandwidth is defined as the number of bits required per repaired bit. 
\end{table}

\begin{theorem}\label{sp2}
	The sub-packetization level of the constructed codes in Construction~2 is $[E:\fq]=\prod_{i=1}^{l}p_i$, where the $p_i$'s are the smallest distinct primes satisfying~\eqref{con_p}, $l$ is the number of exclusion sets of evaluation points.
\end{theorem}

\begin{remark}\label{sp_compare_2}
	According to Theorem~\ref{sp2}, if the dimension $k$ is strictly greater than $l+1$, then the sub-packetization level of RS codes in Construction 2 is strictly less than the lower bound of sub-packetization level for conventional repair schemes.
\end{remark}

\begin{example}
	For the $(17,9)$ RS code in Example \ref{example2} that is generated by Construction 2, the achieved sub-packetization level is $30$, while the lower bound of the sub-packetization level of $(17,9)$ RS codes under conventional repair schemes is $L=9699690$. 
\end{example}

\section{Conclusions}\label{conclusion}
In this manuscript, we introduce a PE repair scheme for scalar linear codes by introducing a novel parameter termed the  flexibility $t$, where $t\geq 1$. The conventional repair scheme corresponds to the case $t=1$, in which it has been shown that a huge cost at the sub- packetization level is necessary to attain the cut-set bound for the single-node repair. Traditionally, research efforts have centered on pursuing the conventional repair scheme with repair bandwidth achieving the cut-set bound or less than that of the naive repair approach, especially in a scenario where the repair locality is taken to its maximum value $n-1$. 
Our work identifies an intrinsic trade-off among the flexibility, the smallest possible sub-packetization level, and the minimum normalized repair bandwidth. It demonstrates that pursuing the repair scheme with repair locality $n-1$ for scalar MDS codes at the expense of extremely high sub-packetization level is not a wise strategy. Instead, a more favorable approach lies in selecting suitable flexibility for intermediate points~$1<t<\min\{n,n-k\}$, as these values avoid the exorbitantly high sub-packetization level burden while still maintaining a relatively low repair overhead. The derived trade-off curve provides a practical reference for selecting the optimal
flexibility $t$ tailored to specific scenarios.

In detail, we derive a lower bound of sub-packetization level for scalar MDS linear codes that achieve the cut-set bound for single-node repair under the PE repair framework. Furthermore, we analyze the trade-off between the flexibility, the smallest possible sub-packetization level, and the minimum normalized repair bandwidth.
We also construct two classes of RS codes achieving the cut-set bound for single-node repair under the PE repair framework, with a reasonable sub-packetization level that can be strictly lower than the lower bound of sub-packetization level established for the conventional repair framework \cite{tamo2017optimal,tamo2019repair}. This advantage enables the repair procedure of RS codes with optimal repair bandwidth to run efficiently on standard commercial computers. We provide two explicit RS codes and detail the corresponding repair schemes that achieve the cut-set bound. To validate the feasibility of our constructions, we implement the repair process via Magma.

For further investigation, the current lower bound on sub-packetization level for the PE repair schemes needs to be improved.
Additionally, the investigation of multi-node failure scenarios under the PE repair framework is also a worthwhile topic for further exploration.
\bibliographystyle{IEEEtran}
\bibliography{IEEEabrv,RS}
\end{document}